\documentclass[11pt,a4paper,final]{iopart}

\usepackage{iopams}  
\usepackage{graphicx}
\usepackage[breaklinks=true,colorlinks=true,linkcolor=blue,urlcolor=blue,citecolor=blue]{hyperref}

\expandafter\let\csname equation*\endcsname\relax
\expandafter\let\csname endequation*\endcsname\relax
\usepackage{amsmath}

\usepackage[font=small,labelfont=bf,
   justification=justified,
   format=plain]{caption}
   
\usepackage{amsfonts,amssymb,mathbbol,dsfont}
\usepackage{graphicx,psfrag,color}
\usepackage{dcolumn}
\usepackage{bm}
\usepackage{mathbbol, appendix}
\usepackage{enumitem}
\usepackage{slashbox}
\usepackage{subfigure}
\usepackage{amstext,amsthm,bm}
\usepackage[autostyle]{csquotes}
\usepackage[english]{babel}
\usepackage[T1]{fontenc}
\usepackage{float}
\usepackage[utf8]{inputenc}
\floatstyle{boxed}
\restylefloat{figure}

\newcounter{myalgctr}

\newenvironment{myalg}[1]{
	\vspace{-3mm}
   \refstepcounter{myalgctr}
   \textsc{\begin{center} \textbf{Algorithm \themyalgctr:} #1\end{center}}
   }{\par}  
\numberwithin{myalgctr}{section}

\theoremstyle{plain}
\newtheorem{thm}{Theorem}[section]

\newtheorem{coro}[thm]{Corollary}
\newtheorem{lem}[thm]{Lemma}

\theoremstyle{definition}
\newtheorem{defn}[thm]{Definition}

\theoremstyle{remark}
\newtheorem{remark}[thm]{Remark}

\def\vone{\mathds{C}^N\otimes \mathds{C}^d}

\def\Z{\mathds{Z}}
\def\N{\mathds{N}}
\def\C{\mathds{C}}

\def \Span{{\rm Span}}
\def\Ker{{\rm{Ker}\,}}
\def\Range{{\rm{Range}\,}}

\def\T{{\bm T}}
\def\Tminus{{\bm T^{-1}}}
\def\pn{{\bm P}_N}
\def\P{{\bm P}}
\def\V{{\mathcal{V}}}
\def\A{{\bm A}}

\def\g{{\bm g}}
\def\E{{\bm E}}
\def\D{{\bm D}}
\def\F{{\bm F}}
\def\W{{\cal W}}
\def\Pib{{\bm \Pi}}
\def\H{{\bm H}}
\def\M{{\bm M}}

\begin{document}

\title[]{Exact solution of corner-modified 
banded block-Toeplitz eigensystems}

\author{Emilio Cobanera$^1$, Abhijeet Alase$^1$, Gerardo Ortiz$^{2,3}$, 
and \\
Lorenza Viola$^1$}

\address{$^1$ Department of Physics and Astronomy, Dartmouth 
College, \\ 6127 Wilder Laboratory, Hanover, NH 03755, USA}

\address{$^2$ Department of Physics, Indiana University,
Bloomington, \\ Indiana 47405, USA}

\address{$^3$ Department of Physics, University of Illinois, 1110 W Green Street, \\
Urbana,  Illinois 61801, USA}

\ead{{emilio.cobanera@dartmouth.edu; abhijeet.l.alase.GR@dartmouth.edu;} \\ ortizg@indiana.edu; \\
lorenza.viola@dartmouth.edu}

\begin{abstract}
Motivated by the challenge of seeking a rigorous foundation for the 
bulk-boundary correspondence for free fermions, 
we introduce an algorithm for determining exactly {\em the spectrum 
and a generalized-eigenvector basis} of a class of banded block 
quasi-Toeplitz matrices that we call {\em corner-modified}. Corner 
modifications of otherwise arbitrary banded block-Toeplitz matrices capture 
the effect of boundary conditions and the associated 
breakdown of translational invariance. Our algorithm leverages the interplay 
between a non-standard, projector-based method of kernel determination 
(physically, a bulk-boundary separation) and families of linear 
representations of the algebra of matrix Laurent polynomials.
Thanks to the fact that these representations act on infinite-dimensional 
carrier spaces in which translation symmetry is restored, it becomes possible 
to determine the eigensystem of an auxiliary projected block-Laurent matrix.  This 
results in an {\em analytic eigenvector Ansatz}, independent of the system size,  
which we prove is guaranteed to contain the full 
solution of the original finite-dimensional problem.   
The actual solution is then obtained by imposing compatibility with a 
{\it boundary matrix}, also independent of system size. 
As an application, we show analytically that eigenvectors of 
short-ranged fermionic tight-binding models may 
display {\em power-law corrections} to exponential decay, and 
demonstrate the phenomenon for the paradigmatic Majorana chain of Kitaev. 
\end{abstract}

\vspace{2pc}

\maketitle


\section{Introduction}
Toeplitz matrices, namely, matrices whose entries are constant 
along their diagonals, find widespread applications 
across mathematics, physics and the engineering sciences, largely 
reflecting the fact that they incorporate three key ingredients of model building: 
discrete approximation of continuous parameters, translation symmetry, and locality 
(typically, in space and/or time). 
As a result, spectral problems associated to Toeplitz 
matrices and related modifications -- in particular, {\em banded} and 
{\em block} Toeplitz matrices where, loosely speaking, the Toeplitz property is 
retained within a finite stripe of non-vanishing diagonals, and the repeated elements 
are themselves matrices -- are of interest from both an analytic and a numerical standpoint 
\cite{toeplitzbook,bottcher05,gray}.  
This paper provides an algorithm, along with its proof, for computing in exact 
form the eigensystem (eigenvalues and generalized eigenvectors) of a class of 
{\em banded block quasi-Toeplitz} matrices, where the repeated blocks are 
square and the quasi- (or ``nearly'') Toeplitz nature arises from constraining allowed 
changes to a small number of rows --  resulting in what we call a {\em corner modification}, 
in a sense to be made mathematically precise later.  

Our interest in the eigensystem problem for banded block-Toeplitz (BBT)
matrices is motivated by contemporary problems in the statistical mechanics 
of independent fermions and bosons, which can be solved by mapping their 
second-quantized Hamiltonians to single-particle Hamiltonians \cite{blaizot,abc}. If
the relevant (fermionic or bosonic) degrees of freedom are labelled by 
lattice sites and internal quantum numbers, then the single-particle Hamiltonian is a 
matrix of size determined by the number of lattice sites times the number of internal 
degrees of freedom, squared. Moreover, if the model of interest involves finite-range 
couplings and is translation-invariant {\em up to boundary conditions}, then the 
single-particle Hamiltonian naturally acquires a corner-modified BBT 
structure \cite{abc}.  
While such Hamiltonians are indeed associated to Hermitian matrices for fermions, 
bosons are described by matrices that are self-conjugate with respect 
to a symplectic inner product \cite{blaizot}.  Both for this reason, and with an eye 
toward extensions to non-Hamiltonian dissipative quantum dynamics 
\cite{prosen,kosov}, our interest here lies in general, {\em not} necessarily Hermitian, 
BBT matrices. 

Translation-invariant systems of independent 
fermions can be classified by topological methods \cite{bernevigbook}, 
leading to a main distinction between topologically trivial and non-trivial systems, and 
finer distinctions within the topologically non-trivial class. A peculiar and conceptually 
unsettling feature of the topological classification of free fermions stems from 
acknowledging that, on the one hand, it heavily leverages translation symmetry.
Yet, its cornerstone prediction, the celebrated ``bulk-boundary
correspondence'', mandates the structure of energy eigenstates (and the existence of 
symmetry protected ``boundary modes'') in systems where translation symmetry has 
been broken -- by boundary conditions, in the simplest setting where disorder plays no 
role. But how, exactly, does the structure of a translation-invariant system 
``descend'' onto a broken-symmetry realization of the system?  Searching for a 
quantitative answer to this question, we introduced in Ref.\,\cite{abc}
the essential core of our present diagonalization algorithm. 
As it will become clear, our analysis pinpoints very explicitly the interplay between 
translation invariance and boundary conditions, and the factors that determine the 
precise structure of the eigenstates of clean single-particle Hamiltonians. 

Beyond condensed matter physics, 
BBT matrices are routinely encountered in a variety of contexts within 
dynamical systems and control theory. 
For example, the method of finite differences maps translation-invariant linear differential equations to 
BBT matrices \cite{beam93}, and, conversely, any BBT matrix may be (non-uniquely) associated 
to some such differential equation. A large class of boundary conditions
can then be mapped to corner-modifications 
of the BBT matrix.
Likewise, time-invariant descriptor systems \cite{luenberger78} 
are naturally associated to BBT matrices by recognizing sequences 
as vectors in an infinite-dimensional space. While sometimes such 
models arise from the discretization of a continuous-time 
dynamical system, very often they do not, as 
no underlying differential equation is known or even expected to exist. 
The simplest model of such a dynamical system is a BBT matrix of bandwidth 
two, whose spectral theory has been well investigated by methods akin
to the transfer matrix method of statistical mechanics 
\cite{nikoukhat87,luenberger89,Chua}. 
Recently, special exact results for tridiagonal BBT matrices 
have been derived by carrying out unconventional mappings
to BBT matrices of bandwidth two \cite{fagotti16}. While time-invariant 
descriptor systems are also a natural source of BBT matrices
with non-square blocks, we will not consider this case in the present work.

Our diagonalization algorithm, advanced without proof in Ref.\,\cite{abc}, 
is very closely adapted to the structure of 
corner-modified BBT matrices and has a crucial main strength: 
It computes the basis of generalized eigenvectors analytically. Both in 
order to clarify this point, and to better place it in context, it is useful 
to outline the strategy of our algorithm here. The main task, namely, 
{\em computing the spectrum and basis of generalized
eigenvectors of an arbitrary corner-modified BBT matrix}, 
is accomplished by breaking it down into three auxiliary tasks.  Specifically:

$\bullet$ {\bf Task I:} 
{\it Compute a basis of the kernel of an arbitrary corner-modified BBT 
matrix.}
The solution of this problem is one of our main contributions. 
We provide an efficient algorithm, the {\it kernel algorithm}, 
that implements an unconventional method of kernel determination 
by projectors, in a manner tailored to corner-modified 
BBT matrices.   

$\bullet$  {\bf Task II:} 
{\it Solve the eigenproblem 
associated to an arbitrary corner-modified BBT matrix.}
We do not develop a new algorithm for completing this task
in itself. Rather, we exploit some of the finer properties of our kernel 
algorithm, in particular, the {\em boundary matrix} $B$ \cite{abc}, in order 
to search for the eigenvalues efficiently. The basis of the corresponding 
eigenspaces is obtained as a by-product. 

$\bullet$ {\bf Task III:} 
{\it Find a basis of generalized eigenvectors.}
If \(C\) denotes the corner-modified BBT matrix of interest, 
we take as input for this problem the output of task II. For each eigenvalue 
\(\epsilon\) of \(C\), 
we call on our kernel algorithm and compute the dimension of the kernel of $(C-\epsilon)^2$. 
If the latter is larger than the dimension of the eigenspace
associated to $\epsilon$, then we compute the dimension of the kernel of
$(C -\epsilon)^3)$, and continue till the dimension is stabilized, for some
$\kappa >1$. In order for this to work,
\((C-\epsilon)^\kappa\) must be shown to be a corner-modified 
BBT matrix, and it must be possible to compute it efficiently. 
This problem is solved by our {\em multiplication algorithm}, which 
is another main contribution of this paper.

Most of the theoretical and computational advantages of our 
method of diagonalization can be traced back to the fact that it
identifies an explicit, and small, search space spanned by
an exact basis: This basis is described in terms of elementary 
functions, and the roots of an associated polynomial whose degree 
is independent of the size of $C$. 
For theoretical work, having
access to the exact structure of eigenvectors, and the factors
that determine this structure, is an invaluable and rarely 
available resource. This happens for instance in the well-known 
algebraic Bethe-Ansatz approach \cite{sutherland}. 
From a computational point of view, it means
that extremely large matrices can be cast in Jordan normal form
without ever storing the matrix or eigenvectors. The actual numerical
task consists of computing the roots of a polynomial of relatively
low degree, parametrically in one complex variable \(\epsilon\), 
and determining the kernel of the relatively small boundary matrix $B$. 
This small kernel determines the linear combinations of potentially
extremely large (but describable in terms of elementary functions) 
vectors which form a generalized-eigenvector basis of the BBT matrix under 
investigation.

Our algorithm includes as a special but conspicuous instance
a seminal algorithm by Trench for computing the eigensystem of 
{\em banded but non-block} Toeplitz matrices \cite{Trench}, and
an important generalization of Trench's work by Beam and Warming 
\cite{beam93}. 
A closed-form solution of the eigensystem problem for tridiagonal 
Toeplitz matrices, illustrating these authors' ideas, can be found in 
Ref.\,\cite{yueh05}. Trench's investigations triggered focused 
interest in the eigensystem problem for banded Toeplitz matrices, 
to the point that a book devoted to the subject has appeared some 
ten years ago \cite{bottcher05}, exclusively about the non-block case. 
Somewhat complementary, an excellent summary of known results on 
{\em block} Toeplitz matrices can be found in Ref.\,\cite{toeplitzbook}, 
which however does not address the spectral properties of BBT matrices --
except implicitly, as a special case of block-Toeplitz matrices with 
continuous symbols. This divide is illustrative of a general pattern. 
Reference \,\cite{beam93}, for example, recognizes the value of Trench's 
algorithm for finite-difference methods and generalizes it in order 
to handle non-Dirichlet boundary conditions on the lattice. The 
matrices they considered are indeed, in the language of this paper, 
corner-modified Toeplitz matrices, but
they are non-block.
This is also an extremely important special case of our work,
and the first one to explore the delicate numerical nature of the
spectral problem for non-Hermitian banded Toeplitz matrices. For
a recent discussion of this point, see also Ref.\,\cite{bertaccini07}. 

To the best of our knowledge, Trench's original 
algorithm was never developed to cover the block case before this 
paper. The early historical evolution of the subject is illuminated 
by Ref.\,\cite{Bini}, where the eigensystem problem for BBT matrices 
was taken up for the first time. While directly motivated by 
Trench's success in the non-block case, these authors did not 
attempt to follow his line of attack.
Interestingly, they favored instead an 
approach that physicists would recognize as the above-mentioned transfer 
matrix method, which indeed has been often and successfully implemented \cite{lee}. 
However, with contemporary computational resources, 
our algorithm (and Trench's as a special instance) becomes extremely 
effective: it combines, in a way that is optimal for the task, 
physical, mathematical, and procedural insight. 
Let us also mention in passing that a problem that has been much investigated 
in the non-block case is associated to Toeplitz matrices perturbed by impurities 
\cite{bottcher02}. Our work can easily extend this investigations to the block case,
where there is considerable room for surprises from a physical 
perspective \cite{prb1}. 

Following a review of the required linear-algebraic concepts and tools in 
Sec. \ref{sec:back}, our core results are presented in Sec. \ref{secexact}. 
In particular, we show how our main task is equivalent to 
solving an appropriate linear system, consisting of a pair of {\em bulk and boundary 
equations}, and prove how the above-mentioned auxiliary tasks are achieved.           
In Sec. \ref{sec:algorithms}, the approach of kernel determination by projectors 
is translated in an explicit algorithmic procedure, and the corresponding time- 
and space-complexity assessed. Sec. \ref{sec:applications} focuses on physical 
implications. After spelling out the general eigenvalue-dependent Ansatz that our 
approach provides for the eigenvectors of a corner-modified BBT matrix, we 
revisit the paradigmatic Kitaev chain for open boundary conditions \cite{kitaev01}, 
and diagonalize 
it in closed form in a parameter regime not fully characterized to date.  
Remarkably, our work uncovers what seems to be a new result in the 
many-body literature, namely, the possibility of (and reason for) having
{\em power-law zero-energy modes} in short range tight-binding models, 
as we indeed explicitly demonstrate in the Kitaev chain.  We include in the 
Appendix a discussion of the diagonalization problem for {\em infinite} 
corner-modified BBT transformations, as relevant to semi-infinite systems.

\section{Preliminaries}
\label{sec:back}

Throughout the paper, \(\mathds{N}=\{1,2,\dots\}\) shall denote the 
counting integers, \(\Z\) the integers, and \(\C\) the field of 
complex numbers, respectively. The algebra of \(d\times d\) complex 
matrices is denoted by \(\M_d\). We will label vectors in the complex 
Hilbert space under consideration, say, \(\V\), with a greek letter, encased in 
Dirac's ket notation (e.g.,  \(|\psi\rangle\in\V\)). 
Then, the Riesz representation lemma guarantees that the space of 
bounded linear functionals of \(\V\) is in antilinear, bijective, 
canonical correspondence with \(\V\).  The unique linear functional 
associated to \(|\psi\rangle\) is denoted by \(\langle \psi |\), and 
the evaluation of this functional on \(|\phi\rangle\) is denoted by 
\(\langle \psi|\phi\rangle\), which is also the Hermitian inner product 
between the vectors \(|\phi\rangle,|\psi\rangle\). We will avoid Dirac's 
notation if the vector space under consideration is not a Hilbert space. 
If $M$ is a linear transformation of $\mathcal{V}$, and \(\mathcal{W}\)  
a subspace of \(\mathcal{V}\), then we write 
\(M|_{\mathcal{W}}:\mathcal{W}\rightarrow \mathcal{V}\) 
for the map obtained by restricting the action of \(M\) to \(\mathcal{W}\). 

Let \(\epsilon\in \C\) denote an eigenvalue of \(M\). Recall that 
$|\psi\rangle \in \mathcal{V}$ is a {\em generalized eigenvector} of $M$, 
of rank \(\kappa \in \mathds{N}\), if 
\(
|\psi\rangle\in {\rm Ker}\,(M-\epsilon)^\kappa\), with 
\( |\psi\rangle\notin {\rm Ker}\,(M-\epsilon)^{\kappa-1}. \)
A generalized eigenvector of rank \(\kappa=1\) is an eigenvector in 
the usual sense. Since ${\rm Ker}\,(M-\epsilon)^\kappa \subseteq 
{\rm Ker}\,(M-\epsilon)^{\kappa'},$ $\forall \kappa' \ge \kappa$, any 
generalized eigenvector $|\psi\rangle$ of rank $\kappa$ also belongs 
to ${\rm Ker}\,(M-\epsilon)^{\kappa'}$. Hence, one may define 
the {\em generalized eigenspace} of $M$ corresponding to its eigenvalue 
$\epsilon$ as the collection of all {\em generalized kernels},
\[ \mathcal{N}_{M,\epsilon} 
= \bigcup_{\kappa\in\mathds{N}}{\rm Ker}\,(M-\epsilon)^\kappa. \]
If $\V$ is a finite-dimensional complex vector space, then for every 
eigenvalue $\epsilon$ of $M$, there exists a 
$\kappa_{\rm max}\in \mathds{N}$, such that 
\[
{\rm Ker}\,(M-\epsilon)^{\kappa_{\rm max}} 
= {\rm Ker}\,(M-\epsilon)^{\kappa_{\rm max}+1}
 = \mathcal{N}_{M,\epsilon}.
\]
Since \(\V=\bigoplus_{\epsilon}\mathcal{N}_{M,\epsilon}\) \cite{halmos}, 
where the direct sum runs over all the eigenvalues of \(M\), 
a maximal linearly independent subset of the set of all generalized 
eigenvectors of \(M\) yields a basis of generalized eigenvectors
of \(\V\).

\subsection{Corner-modified banded block-Toeplitz matrices}
\label{seccmBBT}

A matrix $A_N\in \M_{dN}$ of size \(dN\times dN\) 
is a {\em block-Toeplitz matrix} if there exists a sequence 
\(\{a_j\in \M_d\}_{j=-N+1}^{N-1}\) of \(d\times d\) matrices
such that 
$[A_N]_{ij}\equiv a_{j-i}$, $1\le i,j \le N$,
as an array of \(d\times d\) blocks. Graphically,  $A_N$ has the structure
\begin{eqnarray}
\label{alltoeplitz}
 A_N=
\begin{bmatrix}
a_0       & a_1         &\cdots   &  a_{N-2}    & a_{N-1}   \\
a_{-1}    & a_0         &\ddots   & \           & a_{N-2}   \\
\vdots    & \ddots      &\ddots   & \ddots      & \vdots    \\
a_{-N+2}  & \           &\ddots   & a_0         & a_{1}     \\
a_{-N+1}    &  a_{-N+2}   &\cdots   & a_{-1}      & a_0
\end{bmatrix}
\ , \quad a_j \in \M_d.  
\end{eqnarray}
A block-Toeplitz matrix is {\em banded} if there exist 
``bandwidth parameters'' $p, q \in \Z$, with $-N+1<p\leq q<N-1$, 
such that \(a_{p},a_q \ne 0\) and $a_r=0$ if $r<p$ or $r>q$. 
Accordingly, the graphical representation of a {\em banded block-Toeplitz} 
(BBT) matrix is
\[
A_N=
\begin{bmatrix}
a_0    & \dots  & a_q    &0       &\cdots  &0      \\
\vdots &\ddots  & \      &\ddots  &\ddots  &\vdots \\ 
a_p    &\       &\ddots  &        &\ddots  & 0     \\
0      &\ddots  &\       &\ddots  &        & a_q    \\
\vdots &\ddots  &\ddots  &\       &\ddots  &\vdots  \\    
0      &\cdots  &0       & a_p    &\cdots  & a_0
\end{bmatrix}\ . 
\]
The entries \(a_p,a_q\) are the {\em leading coefficients} of 
\(A_N\), and the pair of integers \((p,q)\) defines the bandwidth \(q-p+1\). 
The transpose and the adjoint of a BBT matrix of bandwidth \((p,q)\) 
are both BBT matrices of bandwidth \((-q,-p)\).  
Let \(p' \equiv {\rm min}(p,0)\), \(q' \equiv {\rm max}(0,q)\). The 
{\em principal coefficients} \(a_{p'}, a_{q'}\) of \(A_N\) are defined 
by
\begin{eqnarray*}
a_{p'} \equiv \left\{
\begin{array}{lrl}
a_p& \mbox{if} & p\leq 0\\
0  & \mbox{if} & p>0
\end{array}\right. ,\quad
a_{q'}\equiv \left\{
\begin{array}{lrl}
a_q& \mbox{if} & q\geq 0\\
0  & \mbox{if} & q<0
\end{array}\right. .
\end{eqnarray*}
Leading and principal coefficients differ only if \(A_N\) is 
strictly upper or lower triangular. 

Block matrices of size $dN \times dN$ induce linear transformations of 
\({\cal V}=\vone\).
Let $\{|j\rangle, \ j=1,2,\dots,N\}$ and $\{|m\rangle, m=1,2,\dots,d\}$ 
denote the canonical bases of \(\C^N\) and \(\C^d\), respectively. Then, 
\[ 
A_N|\psi\rangle=\sum_{i,j=1}^{N}|i\rangle [A_N]_{ij}|\psi_j\rangle,
\quad \mbox{with} \quad 
|\psi_j\rangle=\sum_{m=1}^d\psi_{jm}|m\rangle\quad\mbox{and}\quad \psi_{jm}\in\C. 
\]
For fixed $N,d$ and bandwidth $(p,q)$, the following projectors 
will play a key role in our discussion:
\begin{defn}
\label{bbprojectors}
The {\em right bulk projector} and {\em right boundary projector} 
are given by 
\begin{eqnarray}
\hspace*{-2.2cm}
P_B^{(p,q)} |j\rangle|m\rangle \equiv 
\left\{ \begin{array}{lc} |j\rangle|m\rangle & 1-p' \leq j \leq N-q',\ 1\le m \le d  \\ 
0 & \text{otherwise} \end{array}\right., 
\quad P_\partial^{(p,q)} \equiv  \mathds{1}-P_B^{(p,q)}. 
\end{eqnarray}
The {\em left bulk projector} and {\em  left boundary projector} 
are similarly defined by \(Q_B \equiv P_B^{(-q,-p)}\), 
\(Q_\partial \equiv \mathds{1}-Q_B\), respectively.
\end{defn}

\noindent 
Usually we will write simply \(P_B,P_\partial\). 
Since \(j=1,\dots,N\) and \(p'={\rm min}(p,0)\), the condition 
\(1-p'\leq j\) is trivially satisfied if \(p\geq 0\). This 
situation corresponds to having an upper-triangular BBT matrix. 
Similarly, 
\(j\le N-q'\) is trivially satisfied if \(q\leq 0\), which 
corresponds to a lower-triangular BBT matrix. It is immediate 
to check that \(\dim\Range P_\partial =d(q'-p')\), and so the 
``bulk'' of a BBT 
matrix is non-empty, that is, \(P_B\neq 0\) only if \(N>q'-p'\). 
This makes \(q'-p'\) one of the key ``length scales'' of the
problem, and so we will introduce a special symbol for it,
\begin{equation}
\tau \equiv q'-p'\geq q-p.
\label{eq:tau}
\end{equation}
The condition for a non-empty bulk, \(N>\tau\), is similar-looking 
to the relationship  \(  2(N-1)>  q-p \) that is part of the 
definition of BBT matrix. It is quite possible for a BBT matrix to have
an empty bulk, especially for small \(N\). Such matrices are outside 
our interest and the scope of our methods. 

\begin{defn}
A {\it corner modification} for bandwidth $(p,q)$
is any block matrix \(W\) such that \(P_BW=0\). A corner modification 
is {\it symmetrical} if, in addition, \(WQ_B=0\). A {\em corner-modified BBT matrix} 
$C$ is any block matrix of the form 
\(C=A_N+W\), with \(A_N\) a BBT matrix of size $dN \times dN$ and bandwidth \((p,q)\), 
and $W$ a corner modification.
\end{defn}

\noindent 
If \(W\) is a symmetrical corner modification for bandwidth \((p,q)\),
then its transpose and hermitian conjugate, $W^T$ and $W^\dagger$ respectively, are both
symmetrical corner modifications
for bandwidth \((-q,-p)\), and the other way around. Symmetrical corner 
modifications do indeed look like ``corner modifications''
in array form, whereas this is not necessarily the case for non-symmetrical 
ones.

\subsection{Banded block-Laurent matrices}
\label{secbbl}

Let \(\{a_j\in \M_d\}_{j\in\Z}\) denote a 
doubly-infinite sequence of \(d\times d\) square matrices. 
A {\em block-Laurent matrix} \(\A\) is a doubly-infinite 
matrix with entries
\(
[\A]_{ij} = a_{j-i}\in \M_d,\ i,j\in\mathds{Z}.
\)
A block-Laurent matrix \(\bm{A} \) is {\em banded} if there 
exist  integers \(p,q\), with \(p\le q\), such that 
\(a_j=0\) if
\(j<p\) or \(j>q\), and \(a_p,a_q\neq 0\). For \(p\leq0\leq q\), 
the array
representation of a {\em banded block Laurent} (BBL) matrix is  
\begin{eqnarray}
\label{blocklaurentmat}
\A=
\begin{bmatrix}
\ddots &        &\ddots  &\ddots  &        &        &        &         \\
       & a_0    & \dots  & a_q    &0       &        &        &         \\
\ddots & \vdots &\ddots  & \      &\ddots  &\ddots  &        &         \\ 
\ddots & a_p    &\       &\ddots  &        &\ddots  & 0      &         \\
       & 0      &\ddots  &\       &\ddots  &        & a_q    &\ddots    \\
       &        &\ddots  &\ddots  &\       &\ddots  &\vdots  &\ddots    \\    
       &        &        &0       & a_p    &\cdots  & a_0    &          \\ 
       &        &        &        & \ddots &\ddots  &        & \ddots 
\end{bmatrix}.
\end{eqnarray} 
The bandwidth $(p,q)$, as well as the leading and principal coefficients, of a
BBL matrix are defined just as for BBT matrices. 

A BBL matrix induces a linear transformation of the space 
\begin{equation}
\V^S_d \equiv \Big\{\{|\psi_j\rangle\}_{j\in\Z}\,\big|\, |\psi_j\rangle\in \C^d,\ \forall j\Big\}
\label{Vsd}
\end{equation}
of vector-valued, doubly-infinite sequences. 
Let us write \(\Psi \equiv \{|\psi_j\rangle\}_{j\in\Z}\) 
(since \(\V^S_d\) is {\em not} a Hilbert space, we do not case \(\Psi\) in a 
ket). Then, for a BBL matrix of bandwidth \((p,q)\),
\begin{eqnarray}
\label{action}
\A\Psi=\{a_p|\psi_{j+p}\rangle+\dots+a_q|\psi_{j+q}\rangle\}_{j\in\Z}=
\Big \{\sum_{r=p}^qa_r|\psi_{j+r}\rangle \Big\}_{j\in\Z}\in \V^S_d
\end{eqnarray}
is the associated BBL transformation. If one pictures \(\Psi\) as a 
doubly-infinite block-column vector, one can think of this equation 
as matrix-vector multiplication. The {\it support} 
of a sequence \(\{|\psi_j\rangle\}_{j\in\Z}\) is finite if the 
sequence vanishes but for finitely many values of \(j\). 
Otherwise, it is infinite. 

In the non-block case where \(d=1\), \(\V_{d=1}^S\) becomes 
the space of scalar sequences. 
There is a natural identification \(\V_d^S \simeq \bigoplus_{m=1}^d\V^S_1\)
that we will use often, namely: 
\begin{equation}
\label{dirsumeasy}
\{|\psi_j\rangle\}_{j\in\Z}=
\sum_{m=1}^d\{\psi_{jm}|m\rangle\}_{j\in \Z}\cong
\bigoplus_{m=1}^d\{\psi_{jm}\}_{j\in \Z},
\end{equation}
with respect to the canonical basis \(\{|m\rangle\}_{m=1}^d\) 
of \(\C^d\). Let us define a multiplication of scalar sequences 
by vectors as 
\[  \V_1^S\times \C^d\ni (\Phi,|\psi\rangle)\mapsto 
\Phi|\psi\rangle=|\psi\rangle\Phi=\{\phi_j|\psi\rangle\}_{j\in \Z} \in \V^S_d. 
\] 
Combining this definition with Eq. (\ref{dirsumeasy}) above, 
we finally obtain the most convenient representation of vector 
sequences and BBL matrices:
\[ \hspace*{-0mm}
\{|\psi_j\rangle\}_{j\in\Z}=\sum_{m=1}^d\{\psi_{jm}\}_{j\in \Z}|m\rangle, \quad 
\A\Psi=\sum_{m',m=1}^d|m'\rangle 
\Big\{\sum_{r=p}^q\psi_{j+r,m}\langle m'|a_r|m\rangle\Big\}_{j\in \Z}.
\]

Unlike BBT transformations, BBL transformations close an associative 
algebra with identity \({\bf 1}\), 
isomorphic to the algebra of {\em matrix Laurent polynomials} \cite{mourrain00}. 
The algebra of complex Laurent polynomials 
\(\C[w,w^{-1}]\) consists of complex polynomials in two variables, an indeterminate 
\(w\) and its inverse \(w^{-1}\), with coefficients in \(\C\). Matrix 
Laurent
polynomials are described similarly, but with coefficients in \(\M_d\).
We will denote a concrete but arbitrary matrix Laurent polynomial as 
\begin{eqnarray}
\label{lmp}
A(w,w^{-1})=\sum_{r=p}^qw^r a_r, \quad a_r\in\M_d,\quad a_p,a_q\neq 0.
\end{eqnarray}
In order to map a matrix Laurent polynomial to a BBL matrix $\A$, it is
convenient to write sequences as formal power series, 
\( \Psi=\{|\psi_j\rangle\}_{j\in\Z}=\sum_{j\in \Z}w^{-j}|\psi_j\rangle, \)
which, together with Eq. (\ref{lmp}), yields 
\[ A(w,w^{-1})\sum_{j\in \Z}w^{-j}|\psi_j\rangle=
\sum_{j\in\Z}w^{-j} \big(a_p|\psi_{j+p}\rangle+\dots+a_{q}|\phi_{j+q}\rangle\big).
\]
Comparing this equation with Eq.\,\eqref{action}, one may regard the 
matrix Laurent polynomial \(A(w,w^{-1})\) as inducing 
a BBL transformation via the following algebra isomorphism:
\begin{equation}
A(w,w^{-1})\mapsto \A \equiv \rho_d(A(w,w^{-1})). 
\label{rhomap}
\end{equation}

Since the algebra of matrix Laurent polynomials is generated by 
\(w\mathds{1}\), \(w^{-1}\mathds{1}\) and \(w^0 a=a\in \M_d\) 
($\mathds{1}\in\M_d$ denotes the identity matrix), the 
algebra of BBL transformations is generated by the corresponding linear
transformations of \(\V^S_d\),  which we denote by \(a\equiv \rho_d(w^0 a)\), 
the left shift \(\T\equiv \rho_d(w\mathds{1})\), and the right shift  
\(\Tminus=\rho_d(w^{-1}\mathds{1})\). 
Explicitly, the effect of these BBL matrices on sequences is~\footnote{
By our conventions, we have implicitly agreed 
to use the same symbols \(\T,\Tminus\) to denote the left and right 
shifts of scalar (\(d=1\)) and vector \((d>1)\) sequences. 
As a consequence, e.g., \(\T(\Phi|\psi\rangle)=(\T\Phi)|\psi\rangle\),
illustrating how \(\T\) may appear in multiples places of an equation with 
meanings determined by its use. 
}
\begin{eqnarray*}
\T\Psi=\rho_d(w\mathds{1})\Psi&=&w\mathds{1}\sum_{j\in \Z}w^{-j}|\psi_j\rangle=
\sum_{j\in\Z}w^{-j}|\psi_{j+1}\rangle,\\
a\Psi=\rho_d(w^0 a)\Psi&=&w^{0} a\sum_{j\in \Z}w^{-j}|\psi_j\rangle=
\sum_{j\in \Z}w^{-j} a|\psi_j\rangle,\\
\Tminus\Psi=\rho_d(w^{-1}\mathds{1})\Psi&=&w^{-1}\mathds{1}\sum_{j\in \Z}w^{-j}|\psi_j\rangle=
\sum_{j\in\Z}w^{-j}|\psi_{j-1}\rangle.
\end{eqnarray*} 
In general, 
\( \rho_d\big(\sum_{r=p}^qa_rw^r\big)=\sum_{r=p}^q a_r \T^r \), where we have 
taken advantage of the properties 
\([a,\T]=0=[a,\Tminus]\) and \(\rho_d(\mathds{1})={\bf 1}\).

The relationship between BBT and BBL transformations 
can be formalized in terms of projectors. For integers 
\(-\infty\leq L\leq R\leq\infty\), let
\begin{equation}
\V_{L,R} \equiv \Big\{\{|\psi_j\rangle\}_{j\in\Z}\in\V^S_d\,\Big|\,
\mbox{\(|\psi_j\rangle=0\) if \(j<L\) or \(j>R\)}\Big\}.
\label{Vlr}
\end{equation}
In particular, \(\V_{-\infty,\infty} =\V^S_d\), as defined in Eq. (\ref{Vsd}). 
One may think of \(\V_{L,R}\) as the range of the projector 
\[
\P_{L,R}\{|\psi_j\rangle\}_{j\in\Z} \equiv \{|\chi_j\rangle\}_{j\in\Z},\quad
|\chi_j\rangle=\left\{
\begin{array}{lcl}
0& \mbox{if}& j<L\\
|\psi_j\rangle & \mbox{if} & j=L,\dots,R\\
0& \mbox{if}& j>R
\end{array}\right. .
\]
If \(2(R-L)>q-p\), the linear transformation 
\( \P_{L,R}\A|_{\V_{L,R}}\) of \(\V_{L,R} \) 
is induced by a BBT matrix \(A_N\), with 
\( N=R-L+1. \)
To see this, write 
\( \A\Psi=\sum_{i,j\in\Z}w^{-i}[\A]_{ij}|\psi_j\rangle
=\sum_{i,j\in\Z}w^{-i}a_{j-i}|\psi_j\rangle. \)
The sum over \(i\) is formal, and the sum over \(j\) exists because
\(a_{j-i}=0\) if \(j-i<p\) or \(j-i>q\). If \(\Psi\in \V_{L,R}\), then 
\[  \P_{L,R}\A\Psi=
\sum_{i,j=L}^Rw^{-i}a_{j-i}|\psi_j\rangle=
\sum_{i,j=1}^Nw^{-(i+L-1)}[A_N]_{ij}|\psi_{j+L-1}\rangle,\]
with \(A_N\) a BBT matrix.  Similarly, \(\P_{L-p',R-q'}\A|_{\V_{L,R}}\) 
is induced by the 
corner-modified BBT matrix \(P_BA_N\). In this sense, we will write
\[ A_N=\P_{L,R}\A|_{\V_{L,R}},\quad\mbox{and}\quad P_BA_N=\P_{L-p',R-q'}\A|_{\V_{L,R}}.
\]
The bulk of the matrix \(A_N\) is non-empty if \(R-L\geq \tau\), as in 
Eq.\,(\ref{eq:tau}).

\subsection{Regularity and the Smith normal form}

A matrix Laurent polynomial may be associated to a standard matrix 
polynomial, involving only non-negative powers. In particular, the 
matrix polynomial (in the variable \(w\)) of the matrix Laurent 
polynomial \(A(w,w^{-1})\) in Eq. (\ref{lmp}) is given by 
\begin{equation}
G(w) \equiv w^{-p}A(w,w^{-1})=\sum_{s=0}^{q-p}w^{s}a_{s+p}.
\label{mp}
\end{equation}
A matrix polynomial is called {\em regular} if its determinant 
is not the zero polynomial. Otherwise, it is {\em singular}.
For example, direct calculation shows that \(G(w)\) in
\[  A(w,w^{-1})=w^{-1}G(w) \equiv 
\begin{bmatrix}
w+w^{-1}-\epsilon& w-w^{-1}\\
-w+w^{-1}& -w-w^{-1}-\epsilon
\end{bmatrix}  \]
is regular unless the parameter \(\epsilon=\pm 2\).
By extension, a matrix Laurent polynomial \(A(w,w^{-1})\) and 
associated BBL matrix \(\rho_d(A(w,w^{-1}))\) are  regular or singular 
according to whether the polynomial factor of \(A(w,w^{-1})\) is regular 
or singular. Finally, a BBT matrix \(A_N=  \P_{L,R}\A|_{\V_{L,R}}\) is 
regular (singular) if \(\A\) is. 

A useful fact from the theory of matrix polynomials is that they can be put 
in {\em Smith normal form} by Gaussian elimination \cite{matrix_polynomials}. 
That is, there exist \(d\times d\) matrix polynomials \(E(w),D(w),F(w)\) such 
that the {\em Smith factorization} holds:
\begin{eqnarray}
G(w)=E(w)D(w)F(w).
\label{smith0}
\end{eqnarray}
Here, \(E(w)\), \(F(w)\) are non-unique matrix polynomials with matrix polynomial inverse, and  
\begin{eqnarray}
\label{smith}
D(w) \equiv 
\begin{bmatrix}
g_{1}(w)& \         &\                 &\ &\      &\ \\ 
\                & \ddots    &\                 &\ &\      &\ \\
\                & \         &g_{d_0}(w) &\ &\      &\ \\ 
\                & \         & \                &0 &\      &\ \\  
\                &\          &\                 &\ &\ddots &\ \\
\                &\          &\                 &\ &\      &0
\end{bmatrix},
\end{eqnarray}
is the unique diagonal matrix polynomial with the property that 
each \(g_{m}(w)\) is monic (i.e., has unit leading coefficient) and 
$g_m(w)$ divides \(g_{m'}(w)\), that is, \(g_m(w)|g_{m'}(w)\), if 
\(1\leq m < m'\leq d_0\leq d\). \(D(w)\) is the {\it Smith normal form} 
of \(G(w)\). The matrix polynomial $G(w)$ is singular if and only if $d_0<d$, 
so that its  Smith normal form $D(w)$ has zeroes on the main diagonal. 
Since, from Eq. (\ref{mp}), $G(0)=a_p$ and \([w^{q-p} G(w^{-1})]_{w=0}=a_q\),
it is easy to check that {\em a BBL (or BBT) matrix with at least one 
invertible leading coefficient is regular}.

The Smith factorization of \(G(w)\) immediately implies the 
factorization \(A(w,w^{-1})=w^{p}E(w)D(w)F(w)\). By combining this 
result with the representation defined in Eq. (\ref{rhomap}), one 
obtains what one might reasonably call the {Smith factorization of a BBL matrix}:
\begin{eqnarray}
\label{snfbbl}
\hspace{-2cm}
\A=\T^{p}\E\D\F,\quad  \mbox{with}\quad  \E=\rho_d(E(w)),\quad 
\D=\rho_d(D(w)), \quad \F=\rho_d(F(w)).
\end{eqnarray} 
By construction, the linear transformations \(\E,\ \F\) of \(\V_d^S\) 
are invertible BBL matrices. The BBL matrix \(\D\) is the Smith normal 
form of \(\A\). In array form,  
\begin{eqnarray}\label{sosimple}
\hspace{-2cm}
\D=\rho_d(D(w))=
\begin{bmatrix}
\g_{1}           &  \        &\                 &\ &\      &\ \\ 
\                & \ddots    &\                 &\ &\      &\ \\
\                & \         &\g_{d_0} &\ &\      &\ \\ 
\                & \         & \                &0 &\      &\ \\  
\                &\          &\                 &\ &\ddots &\ \\
\                &\          &\                 &\ &\      &0
\end{bmatrix},\quad \quad \g_m=\rho_{d=1}(g_m(w)).
\end{eqnarray}
The \(\g_m\) are banded, non-block Laurent matrices: they are linear 
transformations of the space \(\V^S_1\) of scalar sequences.

\section{Structural characterization of kernel properties} 
\label{secexact}

\subsection{The bulk/boundary system of equations}
\label{seckerinclusion}

Our algorithm for computing the kernel of a corner-modified BBT matrix builds on an indirect method for 
determining the kernel of a linear transformation, that we term {\em kernel 
determination by projectors}.  The starting point is the following:

\begin{defn}
Let \(M\) be a linear transformation of \(\V\), and 
\(P_1, P_2 
\equiv \mathds{1}-P_1\) non-trivial projectors, that is, 
neither the zero nor the identity map. The {\em compatibility map} 
of the pair \( (M, P_1) \) is the linear transformation
\(B  \equiv P_1 M|_{\text{Ker}\,P_2M} \).
\end{defn}
The kernel condition \(M|\psi\rangle=0\) is equivalent to the 
system of equations
\( P_1M|\psi\rangle=0$, $P_2M|\psi\rangle=0. \)
In view of the above definition, this means that   
\( \text{Ker}\,M= \text{Ker}\,P_2M\cap \text{Ker}\,P_1M=\text{Ker}\,B. \)
Roughly speaking, since the subspaces \(\text{Ker}\,P_2M\) and 
\(\text{Range}\,P_1\) may be ``much smaller'' than \(\V\), 
it may  be advantageous to determine \( \text{Ker} M\) indirectly, 
by way of its compatibility map. One can make these ideas more precise if \(\dim \V<\infty\).

\begin{lem}
\label{sizebmatrix}
Let $M$ be a linear transformation acting on a finite-dimensional vector space $\cal V$.  Then 
{\rm dim Range} $P_1 \leq$
{\rm dim Ker }$P_2 M \leq$ {\rm dim Ker} $M$ {\rm +} {\rm dim Range} $P_1$.
\end{lem}
\begin{proof}
The dimension of \( \text{Ker} \,M^TP_2= \Ker (P_2M)^T \) is bounded below by 
the dimension of \( \Ker P_2\), which is precisely 
\(\dim\,\text{Range}\,P_1\).
This establishes the first inequality, because, in finite dimension,
the dimension of the kernel of a matrix coincides with that of its 
transpose. For the second inequality, notice that the solutions of 
\(P_2M|\psi\rangle=0\) are of two types: the kernel vectors 
of \(M\), plus the vectors that are mapped by \(M\) into 
\( \text{Ker} P_2\). Hence, the number of linearly independent solutions 
of \(P_2M|\psi\rangle=0\) is upper-bounded by 
\(\dim\,\text{Ker}\, M+ \dim\,\text{Range}\, P_1\), as claimed. 
\end{proof}

It is instructive to note that, 
if dim Ker $P_2 M >$ dim Range $P_1$, then Ker $M$ is necessarily non-trivial.
Since \(\dim\,\text{Range}\, P_1 \equiv n_{P_1}\) determines the number of rows of 
the matrix of the compatibility map (the {\em compatibility matrix} from now on), 
and \(\dim\,\text{Ker}\, P_2M \equiv n_{P_2}\) determines its number of columns, 
this condition implies that the compatibility matrix is 
an \(n_{P_1}\times n_{P_2}\) {\em rectangular} matrix with more columns than rows. 
As a consequence, its kernel is necessarily non-trivial. 
The following general property is also worth noting, for later use.  
Suppose that \(M' \equiv M+W\) and \(P_2W=0\). 
Then, \(W=P_1W\), and \(\text{Ker}\,P_2M'=\text{Ker}\,P_2M.\)
Moreover, 
\[ B' \equiv P_1M'|_{\text{Ker}\,P_2M'}=
P_1M|_{\text{Ker}\,P_2M}+W|_{\text{Ker}\,P_2M}=B+W|_{\text{Ker}\,P_2M}. 
\]

Let now \(C \equiv A_N+W\) denote a \(dN\times dN\) corner-modified 
BBT matrix, with associated boundary projector 
\(P_1 \equiv P_\partial \), and bulk projector 
$P_2 \equiv P_B=\mathds{1}-P_\partial$. 
The task is to compute \(\Ker C\), which coincides with 
the kernel of the compatibility map 
\begin{eqnarray}
\label{boundarymap}
B=P_\partial C|_{\Ker P_B C}=(P_\partial A_N+W)|_{\Ker P_B A_N}.
\end{eqnarray}
As anticipated in Ref.\,\cite{abc}, this approach translates into 
solving a bulk/boundary system of equations. To make this connection 
explicit, let us introduce an index \(b\) such that 
\begin{eqnarray}
\label{rangeb}
\hspace*{-20mm}
{\rm Range}\,P_\partial \equiv {\rm Span}\,\{ |b\rangle |m\rangle\, |\, 
b=1,\dots,-p', N-q'+1,\dots,N;\, m=1,\dots, d\} , 
\end{eqnarray} 
where if \(p'=0\) or \(q'=0\), the corresponding subset of vectors is empty.
Accordingly, recalling Eq. (\ref{eq:tau}),  
\( n_\partial\equiv \dim{\rm Range}\,P_\partial= d\tau. \)
In addition, let 
\begin{eqnarray}
\label{basispban}
\mathcal{B} \equiv 
\{|\psi_s\rangle\,|\, 
s=1,\dots,n_B\equiv\dim\Ker P_BA_N  \}
\end{eqnarray} 
denote a fixed but arbitrary basis of \(\Ker P_BA_N\).
\begin{defn}
\label{boundmat}
The {\it bulk equation} is the kernel equation \(P_B A_N|\psi\rangle=0\).
The {\it boundary matrix} is the \(n_\partial\times n_B\) block-matrix 
\([B]_{bs}\equiv \langle b|B|\psi_s\rangle\), of block-size $d\times 1$. The 
{\it boundary equation} is the (right) kernel equation for \(B\). 
\end{defn}

Thus, one can set up the boundary equation only 
{\em after} solving the bulk equation. Let us show explicitly how
a solution of the boundary equation determines a basis of \(\Ker C = \Ker (A_N+W)\).
With respect to the bases described in Eqs.\,\eqref{rangeb}-\eqref{basispban}, 
the boundary matrix is related to the compatibility
map of Eq.\,\eqref{boundarymap} as  
\( B|\psi_s\rangle=\sum_b|b\rangle [B]_{bs}. \)
Let $|\epsilon_k\rangle= \sum_{s=1}^{n_B}\alpha_{ks} |\psi_s\rangle. $ 
It then follows that 
\[ B|\epsilon_k\rangle=\sum_{b}|b\rangle\sum_{s=1}^{n_B}[B]_{bs}\alpha_{ks}=0 
\quad \Leftrightarrow 
\quad \sum_{s=1}^{n_B}[B]_{bs}\alpha_{ks}=0.
\]
We conclude that \( \{|\epsilon_{k}\rangle\}_{k=1}^{n_C}\) is a 
basis of \(\Ker C\), for some $n_C\leq n_B$, 
if and only if the column vectors of complex coefficients
\( {\bm \alpha}_k  = \big[\alpha_{k 1} \dots \alpha_{k n_B}\big]^{\rm T} \)
constitute a basis of the (right) kernel of $B$. 

In summary, \(\Ker C\) is fully encoded in two pieces of information: 
a basis for the solution space of the bulk equation and 
the boundary matrix $B$. 
Is this encoding advantageous from a computational perspective? There are 
three factors to consider: 

\begin{itemize}
\item[(i)] How difficult is it to solve the bulk equation and store the solution as an explicit basis.

\item[(ii)] How difficult is it to multiply the vectors in this basis by the matrix
\(P_\partial C\).

\item[(iii)] How hard is it to solve the boundary equation.
\end{itemize}

The remarkable answer to (i), which prompted our work in Ref.\,\cite{abc}, 
is that it is easy to compute and store a basis of \(\Ker P_BA_N\). More 
precisely, the complexity of this task is  {\it independent of N}. The 
reason is that most (or even all) of the solutions of the bulk equation 
\(P_BA_N|\psi\rangle=0\) are obtained by determining the kernel for the 
associated BBL transformation \(\A\), such that \(A_N=\P_{L,R}\A|_{\V_{L,R}}\). 
The latter kernel may be described exactly, in terms of elementary
functions and the roots of a polynomial of degree {\em at most} 
\(d(q-p)\), as we prove formally in Theorem \ref{softy}. 
The answer to (ii) depends on \(W=C-A_N\). In order to compute 
the boundary matrix, it is necessary to multiply a basis of the 
solution space to the bulk equation by the matrix 
\(P_\partial C=P_\partial A_N+W\). The cost of this task is 
independent of \(N\) if $W$ is a symmetrical 
corner modification -- which, fortunately, is most often the 
case in applications. Otherwise, the cost of computing the 
boundary matrix is roughly \({\cal O}(N)\). 
Lastly, $B$ is not a structured matrix, thus the answer 
to (iii) boils down to whether it is small enough to handle efficiently 
with standard routines of kernel determination. Generically, if the target 
BBT matrix is regular, we shall prove later [see Theorem \ref{corboundmat}] 
that $B$ is necessarily {\em square}, of size 
\(n_\partial\times n_\partial\). In particular, its size does not grow 
with \(N\). 

Next we will state and prove the main technical results of this 
section, both concerning the bulk equation. Let us begin with an 
important definition:
\begin{defn}
\label{mlr}
Let \(A_N=\P_{L,R}\A|_{\V_{L,R}}\) denote a BBT transformation 
of bandwidth \((p,q)\). Its {\it bulk solution space} is 
\begin{eqnarray}
\label{defmlr}
\mathcal{M}_{L,R}(\A)\equiv\Ker \P_{L-p',R-q'}\A|_{\V_{L,R}}=\Ker P_BA_N.
\end{eqnarray}
\end{defn}
\smallskip

Notice that \( \T^{-n}\mathcal{M}_{L,R}(\A)=\mathcal{M}_{L+n,R+n}(\A),\) 
and that the spaces 
\( \mathcal{M}_{L,\infty}(\A)=\Ker\P_{L-p',\infty}\A\) and 
\( \mathcal{M}_{-\infty,R}(\A)=\Ker\P_{-\infty,R-q'}\A \)
should not be regarded as ``limits'' of \(\mathcal{M}_{L,R}(\A)\).  
We will usually write simply \(\mathcal{M}_{L,R}\) if \(\A\) is 
fixed.  We then have: 
  
\begin{thm}
\label{softy}
Let \(A_N=\P_{L,R}\A|_{\V_{L,R}}\) denote a BBT transformation 
of bandwidth \((p,q)\) and non-empty bulk. Then, 
\begin{enumerate}
\item 
\(  \P_{L,R}\,\Ker \A  \subseteq \mathcal{M}_{L,R}   \)
\item
If the principal coefficients of \(A_N\) are invertible, then 
\( \P_{L,R}\,\Ker\A = \mathcal{M}_{L,R}.  \) 
\end{enumerate}
\end{thm}
\begin{proof}
{\it (i)} We will prove a stronger result that highlights the 
usefulness of the notion of bulk solution space. For any 
\(\widetilde{L},\widetilde{R}\) such that 
\(-\infty\leq \widetilde{L}\leq L\) and \(R\leq \widetilde{R}\leq \infty\), 
the following {\em nesting property} holds:
\begin{eqnarray}
\label{nesting}
\P_{L,R}\mathcal{M}_{\widetilde{L},\widetilde{R}}\subseteq \mathcal{M}_{L,R}
\end{eqnarray} 
In particular,  \(
\P_{L,R}\mathcal{M}_{-\infty,\infty}=\P_{L,R}\Ker \A\subset \mathcal{M}_{L,R}=\Ker P_BA_N
\) establishes our claim.  
By definition, the sequences in \(\mathcal{M}_{L,R}\) are sequences 
in \(\V_{L,R}\) annihilated by \(\P_{L-p',R-q'}\A\). Hence,
we can prove Eq.\,\eqref{nesting} by showing that
\(\P_{L-p',R-q'}\A\P_{L,R}\) annihilates \(\mathcal{M}_{\widetilde{L},\widetilde{R}}\).
To begin with, note that 
\[ \P_{L-p',R-q'}\A\P_{L,R}=\P_{L-p',R-q'}\sum_{r=p}^qa_r\T^r\P_{L,R}
=\sum_{r=p}^q\P_{L-p',R-q'}\P_{L-r,R-r}a_r\T^r. \]
Since \(p'={\rm min}(p,0)\), \(q'={\rm max}(0,q)\), 
\(L-r\leq L-p\leq L-p'\), and $R-r\geq R-q\geq R-q',$ $\forall r=p,\dots,q$,
thus \( \P_{L-p',R-q'}\P_{L-r,R-r}=\P_{L-p',R-q'}\). It follows 
that  $\P_{L-p',R-q'}\A\P_{L,R}=\P_{L-p',R-q'}\A$.
In particular, 
\(\P_{L-p',R-q'}\A\P_{L,R}=\P_{L-p',R-q'}(\P_{\tilde{L}-p',\tilde{R}-q'}\A)\),
which makes it explicit that \(\P_{L-p',R-q'}\A\P_{L,R}\) 
annihilates \(\mathcal{M}_{\widetilde{L},\widetilde{R}}\).  

{\it (ii)}  It suffices to show that \(\Ker P_BA_N\subseteq \P_{L,R} \Ker \A|_{\V_{L,R}}\) 
if  the principal coefficients are invertible. The principal 
coefficients are invertible if and only if \(p\leq0\leq q\) and the 
leading coefficients \(a_p,a_q\) are invertible. Then, since in 
this case  
\[
P_B A_N=
\begin{bmatrix}
0          & \cdots   &\         &\           &\cdots    & 0       & \cdots     & 0      \\ 
\vdots     &          &\         &\           &\         & \vdots  & \          & \vdots \\
0          & \cdots   &\         &\           &\cdots    & 0       & \cdots     & 0      \\
a_p        & \cdots   &a_0       & \cdots     &a_q       & 0        & \cdots    & 0       \\
0          & \ddots   & \        & \ddots     & \        & \ddots   & \ddots    & \vdots  \\
\vdots     & \ddots   & \ddots   & \          & \ddots   & \        & \ddots    & 0      \\
0          &\cdots    & 0        & a_p        &\cdots    &a_0       & \cdots    &a_q    \\
0          & \cdots   & 0        & \cdots     &\         & \        &\cdots     & 0      \\
\vdots     & \        & \vdots   & \          &\         & \        &\          & \vdots  \\
0          & \cdots   & 0        & \cdots     &\         & \        & \cdots    & 0      \\
\end{bmatrix}\ ,\]
it follows that a state \(\Psi \in \Ker P_BA_N\) satisfies  
\[  P_BA_N
\begin{bmatrix}
|\psi_L\rangle\\
\vdots\\
|\psi_R\rangle
\end{bmatrix}
=
\begin{bmatrix}
0\\
\vdots\\
0\\
a_p|\psi_L\rangle+\dots+a_q|\psi_{L+q-p}\rangle\\
\vdots\\
a_p|\psi_{R-q+p}\rangle+\dots+a_q|\psi_{R}\rangle\\
0\\
\vdots\\
0
\end{bmatrix}=0\ ,
\]
using the notation 
\( \Psi=\P_{L,R}\Psi=\{|\psi_j\rangle\}_{j=L}^R= [ |\psi_L\rangle\: \dots \: |\psi_R\rangle ]^T \). 
Thus, \(\Psi\) can be uniquely extended to yield 
a sequence \(\Psi'\in \Ker \A\).  
Compute \(|\psi_{L-1}\rangle\) (the \(L-1\) entry of \(\Psi'\)) as 
\( |\psi_{L-1}\rangle=-a_p^{-1}(a_{p+1}|\psi_L\rangle+\dots+a_q|\psi_{L+q-p-1}\rangle),
\)
and repeat the process to obtain \(|\psi_{L-j}\rangle\) 
for all \(j\geq 1\). Similarly, compute \(|\psi_{R+1}\rangle\) as
\( |\psi_{R+1}\rangle=
-a_q^{-1}(a_{p}|\psi_{R-(q-p)+1}\rangle+\dots+a_{q-1}|\psi_{R}\rangle), \)
and repeat in order to compute \(|\psi_{R+j}\rangle\) for all \(j\geq 1\).
\end{proof}

\begin{thm}
\label{lemker}
If \(A_N=\P_{L,R}\A|_{\V_{L,R}}\) 
is regular with non-empty bulk,  \(\dim \mathcal{M}_{L,R}= d\tau\). 
\end{thm}
\begin{proof}  

Since \(\dim\mathcal{M}_{L,R}=\dim\Ker P_BA_N=\dim \Ker A_N^\dagger P_B\), 
we can focus on keeping track of the linearly independent 
solutions of \(\langle \phi|P_BA_N=0\). First, there are 
the boundary vectors \(\langle\phi|P_\partial=\langle\phi|\).
From the definition of $P_\partial$, there are
precisely \(d\tau\) such solutions, showing that
\(\dim \mathcal{M}_{L,R}\geq d\tau\). Suppose  that \(\dim \mathcal{M}_{L,R}>d\tau\). 
We will show that \(\A\) must then be singular. Let \(L=1\) and \(R=N\) 
for simplicity. By assumption, there exists a nonzero 
\(  \langle\phi|=\langle \phi|P_B=\sum_{j=-p'+1}^{N-q'}\langle j|\langle\phi_j|\neq 0, \)
such that \(\langle \phi|P_BA_N=\langle\phi|A_N=0\).  
Let \(T^r=\P_{1,N} \T^r|_{\V_{1,N}}\). Then, 
\[
0=\langle\phi|A_N=
\sum_{j=-p'+1}^{N-q'}\langle j|\langle\phi_j|
\sum_{r=p}^qa_rT^r
=\sum_{j=-p'+1}^{N-q'}\sum_{r=p}^q\langle j+r|\langle\phi_j|a_r.
\]
It is useful to rearrange the above equation as 
\begin{align}
&0=\langle N-q'+q|\langle\phi_{N-q'}|a_q
+\langle N-q'+q-1|\big(\langle\phi_{N-q'-1}|a_q+\langle\phi_{N-q'}|a_{q-1}\big)+
\nonumber \\
&+\dots
+\langle 2-p'+p|\big(\langle \phi_{1-p'}|a_{p+1}+\langle \phi_{2-p'}|a_{p}\big)
+\langle 1-p'+p|\langle\phi_{1-p'}|a_{p},
\label{precisely}
\end{align}
where all the labels are consistent because \( -N+1<p'\leq p\leq q\leq q'<N+1.\)
Since, by assumption, \(\{\langle \phi_j|\}_{j=1-p'}^{N-q'}\) is not
the zero sequence, the vector polynomial 
\( \langle \phi(w)|\equiv\sum_{j=-p'+1}^{N-q'}w^j\langle\phi_j|\)
is not the zero vector polynomial. It is immediate to check 
that \(\langle \phi(w)|A(w,w^{-1})=0\), as this equation
induces precisely the same relations among the 
\(\{\langle \phi_j\}_{j=1-p'}^{N-q'}\) as Eq.\,\eqref{precisely} does. 
The claim follows if we can show that this 
implies \(\det A(w,w^{-1})=0\). Consider the Smith decomposition
\(A(w,w^{-1})=w^pE(w)D(w)F(w)\). A nonzero vector polynomial cannot
be annihilated by an invertible matrix polynomial. Hence, 
\(\langle \psi(w)|\equiv\langle \phi(w)|E(w)^{-1}\) is a nonzero 
vector polynomial annihilated by \(D(w)\), \(\langle \psi(w)|D(w)=0\). 
Since \(D(w)\) is diagonal, this is only possible if at least one 
of the entries on its main diagonal vanish, contradicting 
the assumption that \(\A\) is regular.
\end{proof}

\subsection{Exact solution of the bulk equation}

In this section we focus on solving the bulk equation, \(P_BA_N|\psi\rangle=0\). 
There are two types of solutions. Solutions with {\em extended support} 
are associated with kernel vectors of the BBL matrix \(\A\) such that 
\(A_N=\P_{L,R}\A|_{\V_{L,R}}\), as implied by Theorem \ref{softy}.  
From a physical standpoint, it is interesting to observe that these solutions 
can be constructed to be translation-invariant, since \([\A,\T]=0\). Any other 
solution, necessarily of {\em finite support} as we will show, may be thought 
of as emergent: these solutions exist only because of the translation-symmetry-breaking 
projection that leads from the infinite system \(\A\) to the finite system
\(A_N\), and {\em only if} the principal coefficients are {\em not} invertible.

\subsubsection{Extended-support solutions.}
\label{infiniteproblem}

In order to determine the kernel of an arbitrary BBL 
transformation, we first establish 
a few results concerning the kernel of a special 
subclass of BBL transformations. Given a non-negative 
integer \(v\) and \(j\in \Z\), let  
\[  j^{(v)} \equiv \left\{
\begin{array}{lcl}
1                     & \mbox{if}  & v=0 ,\\
(j-v+1)(j-v+2)\dots j & \mbox{if}  & v=1,2,\dots
\end{array}\right. \]

\begin{lem}
\label{scalarker}
The family of scalar sequences defined by 
\begin{equation*}
\Phi_{z,0}\equiv 0,\quad \Phi_{z,1} \equiv \{z^{j}\}_{j\in\Z},\quad
\Phi_{z,v}\equiv \frac{d^{v-1}\Phi_{z,1}}{dz^{v-1}}=\{j^{(v-1)}z^{j-v+1}\}_{j\in\Z} ,
\quad v=2,3,\dots, 
\end{equation*}
satisfies the following properties:
\begin{enumerate}
\item 
\(\T\Phi_{z,v}=z \Phi_{z,v}+(v-1)\Phi_{z,v-1}\).
\item
\({\cal S}_{z,s}\equiv \Span\{\Phi_{z,v}\}_{v=1}^s=\Ker (\T-z)^s\), for all $s\in {\N}$.  
\item
For any \(s_1,s_2 \in \N\), if \(z_1\neq z_2\), then 
\({\cal S}_{z_1,s_1}\cap {\cal S}_{z_2,s_2}=\{0\}\). 
\end{enumerate}
\end{lem}

\begin{proof}
{\it (i)} It is immediate to verify that \(\T\Phi_{z,1}=z\Phi_{z,1}\).
Hence, it also follows that
\[\T\Phi_{z,v}
=\frac{d^{v-1}}{dz^{v-1}}\T\Phi_{z,1}
=\frac{d^{v-1}}{dz^{v-1}}(z\Phi_{z,1})
=z\Phi_{z,v}+(v-1)\Phi_{z,v-1}. \]

{\it (ii)} Let \(K_z\Phi \equiv \{z^j\phi_j\}_{j\in\Z}\) for \(z\neq 0\). 
Then, 
\begin{align*}
K_{z}(\T-{\bf 1})\Phi=\sum_{j\in\mathds{Z}}w^{-j}z^j(\phi_{j+1}-\phi_j)=
(z^{-1}\T-{\bf 1})K_z\Phi.
\end{align*}
In other words, \(zK_z(\T-{\bf 1})=(\T-z)K_z\), and so
\(
(\T-z)^{s} K_z=z^s K_{z}(\T-{\bf 1})^{s}.\) 
As a consequence, 
\( 
\Ker(\T-z)^s=K_z\,\Ker(\T-{\bf 1})^s. 
\) 
In particular, for 
\(z=1\), \( \Ker (\T-{\bf 1})\) is spanned by the constant 
sequence \(\Phi_{1,1}=\{\phi_j=1\}_{j\in\Z}\).
Suppose that 
\(
{\cal S}_{1,s}= \Ker(\T-{\bf 1})^s.  
\)
We will proceed by induction, and prove that, as a consequence, 
\( 
{\cal S}_{1,s+1} = {\rm Ker}\,(\T-{\bf 1})^{s+1}. 
\) 
Every $\Psi \in \Ker(\T-{\bf 1})^{s+1}$ satisfies 
\(
(\T-{\bf 1})^{s}(\T-{\bf 1})\Psi=0. 
\)
Then, by the induction hypothesis, $(\T-{\bf 1})\Psi \in {\cal S}_{1,s}$ 
and so there exists numbers \(\alpha_v\) such that
\(
(\T-{\bf 1})\Psi = \sum_{v=1}^{s}\alpha_v \Phi_{1,v}.
\)
Since $(\T-{\bf 1})\Phi_{1,v}=(v-1)\Phi_{1,v-1}$, it follows that
\(
\Psi =\Psi'+ \sum_{v=1}^{s}\frac{\alpha_v}{v} \Phi_{1,v+1},\) with
\( \Psi'\in \Ker(\T-{\bf 1}).  \) 
This shows that $\Psi \in {\cal S}_{1,s+1}$, 
and so  $\Ker(\T-{\bf 1})^{s+1} \subseteq {\cal S}_{1,s+1}.$ 
The opposite inclusion 
holds because \({\cal S}_{1,s}=\Ker(\T-{\bf 1})^s\subset \Ker(\T-{\bf 1})^{s+1}\)
by the induction hypothesis, and 
\(   (\T-{\bf 1})^{s+1}\Phi_{1,s+1} = s(\T-{\bf 1})^{s}\Phi_{1,s} = \dots =s!(\T-{\bf 1})\Phi_{1,1}= 0. \)

{\it (iii)} We will prove a slightly more general result:
Let \({\bm Y}_i\equiv \rho(y_i(w)),$ $i=1,2,\) denote two upper-triangular
banded Laurent transformations of \(\V_1^S\), where without loss of 
generality \(\deg y_2\geq \deg y_1\).  We claim that if \({\rm gcd}(y_1,y_2)=1\), 
then \(\Ker {\bm Y}_1\cap \Ker {\bm Y}_2=\{0\}\).
To prove this, suppose \(\Phi\in \Ker {\bm Y}_1\cap \Ker {\bm Y}_2\). 
The banded Laurent transformation \(\rho(r)\) induced by the remainder 
of dividing \(y_2\) by \(y_1\) also annihilates \(\Phi\), since
\(y_2= c y_1+r\). Continuing this process, we conclude that 
${\rm gcd}(y_1,y_2)$ annihilates $\Phi$. Hence, if 
\({\rm gcd}(y_1,y_2)=1\), then $\Phi=0$.
\end{proof}

\begin{lem}
\label{sdamaps}
The finite-dimensional space of vector sequences
\[ {\cal T}_{z,s}
=\Ker(\T-z)^s
\equiv \Span \{\Phi_{z,v}|m\rangle\, | \, v=1,\dots,s;\ m=1,\dots,d\}
\]
is an invariant subspace of the algebra of BBL matrices. The mapping 
\begin{equation}
A(w,w^{-1})\mapsto \A|_{{\cal T}_{z,s}} \equiv  A_s(z) 
\label{evalmap}
\end{equation}
defines a \(ds\)-dimensional representation of the algebra of matrix Laurent 
polynomials, where the block matrix 
\begin{equation}
[A_s(z)]_{xv} \equiv \left \{
\begin{array}{lcl}
\binom{v-1}{x-1}A^{(v-x)}(z,z^{-1})& \mbox{if} &1\leq x\leq v\leq s \\
0& \mbox{if} & 1\leq v<x\leq s 
\end{array}\right. ,
\quad  A^{(v-x)} \equiv \frac{d^{v-x}A}{dz^{v-x}},
\label{Azs}
\end{equation}
is the matrix of \(A_s(z) \) relative to the above defining basis
of \({\cal T}_{z,s}\), with \(A^{(0)} \equiv A = A_1 \), and $\binom{v}{x}$
denoting the binomial coefficient.
\end{lem}
\begin{proof}
Since   
\( \T\Psi=\sum_{m=1}^d|m\rangle (\T\{\psi_{jm}\}_{j\in\Z}), \)
it is immediate to check that \({\cal T}_{z,s}=\Ker (\T-z)^s\).
Moreover, \((\T-z)^s \A{\cal T}_{z,s}=\A(\T-z)^s{\cal T}_{z,s}=\{0\}\). 
This proves that \(\A{\cal T}_{z,s}\subseteq {\cal T}_{z,s}\). The matrix 
of \(\A|_{{\cal T}_{z,s}}\) is computed as follows. On one hand, by definition, 
\[ \A\Phi_{z,v}|m\rangle= A_s(z) \Phi_{z,v}|m\rangle=
\sum_{x=1}^{s}\sum_{m'=1}^d\Phi_{z,x}|m'\rangle\langle m'| [A_s(z)]_{xv}|m\rangle. \]
On the other hand,
\[ \A\Phi_{z,v}|m\rangle=
\frac{d^{v-1}}{dz^{v-1}}\Phi_{z,1}A(z,z^{-1})|m\rangle=
\sum_{x=1}^v\sum_{m'=1}^d
\Phi_{z,x}|m'\rangle\binom{v-1}{x-1}\langle m'|A^{(v-x)}(z,z^{-1})|m\rangle,\]
where we have taken advantage of 
\(\A\Phi_{z,1}|m\rangle =\sum_{r=p}^q(\T^r\Phi_{z,1})a_r|m\rangle=A(z,z^{-1})|m\rangle. \) 
\end{proof}

\noindent
We call the finite-dimensional representation 
defined in Eqs. (\ref{evalmap})-(\ref{Azs}) 
the {\it generalized evaluation map at \(z\) of degree \(s\)}, 
with \( A_1(z) \equiv A(z,z^{-1})\) recovering the usual 
evaluation map of polynomials. Explicitly, we have 
\begin{equation}
\label{AzsBlock}
A_s(z)=
\begin{bmatrix}
A      & A^{(1)}      & A^{(2)}    &\dots       &A^{(s-1)} \\
0      & A            & 2 A^{(1)}  &\ddots      & \vdots \\
\vdots & \ddots       & \ddots     &\ddots      &\frac{(s-1)(s-2)}{2}A^{(2)}\\
\vdots &              &\ddots      & \ddots     & (s-1)A^{(1)}\\
0     &  \dots       & \dots       &0           & A
\end{bmatrix}, 
\end{equation} 
\begin{align*}
\A\sum_{v=1}^{s}\Phi_{z,v}|\psi_v\rangle = 
\Big[ \Phi_{z,1} \: \Phi_{z,2} \: \dots \: \Phi_{z,s-1} \: \Phi_{z,s} \Big]
A_s(z) 
\begin{bmatrix}
|\psi_1\rangle    \\
|\psi_2\rangle    \\
\vdots         \\
|\psi_{s-1}\rangle\\
|\psi_s\rangle
\end{bmatrix}.
\end{align*}
It is instructive to verify the representation property 
explicitly in an example. Direct calculation shows that
\[ (w\mathds{1})_3 (z)=
\begin{bmatrix}
z\mathds{1} & \mathds{1}  & 0\\
0           & z\mathds{1} & 2\mathds{1}\\
0           & 0           & z\mathds{1}
\end{bmatrix}
\quad \mbox{and}\quad 
(w^{-1}\mathds{1})_3(z)=
\begin{bmatrix}
\mathds{1}/z & -\mathds{1}/z^2 &  2\mathds{1}/z^3\\
0            & \mathds{1}/z   & -2\mathds{1}/z^2\\
0            &  0              &  \mathds{1}/z
\end{bmatrix}. 
\]
It is then immediate to check that 
\[ (w\mathds{1})_3 (z)  \, (w^{-1}\mathds{1})_3 (z)=
(w^{-1}\mathds{1})_3(z) \,  (w\mathds{1})_3(z) = (\mathds{1})_3(z) = \mathds{1}_{3d} .\] 

\vspace*{1mm}

The Smith decomposition of \(\A\), Eq.\,\eqref{snfbbl}, implies that
\( \Ker\A=\F^{-1}\Ker\D, \) 
where \(\F\) is a regular BBL matrix and the Smith normal 
form \(\D\) of \(\A\) has the simple structure shown in Eq.\,\eqref{sosimple}.  
In this way, the problem of determining a basis of \(\Ker\A\) 
reduces to two independent tasks: determining a basis 
of \(\Ker \D\), and the change of basis \(\F^{-1}\).  
With regard to \(\Ker \D\), its structure may be characterized 
quite simply, after introducing some notation. 
Let \(\{z_{\ell}\}_{\ell=0}^n \) 
denote the {\em distinct} roots of \(g_{d_0}(w)=D_{d_0d_0}(w)\).
By convention, \(z_0 \equiv 0\). 
Since every \(g_m(w)=D_{mm}(w)\) divides \(g_{d_0}(w)\), it must 
be that the  roots of the \(g_m(w)\) are also
roots of \(g_{d_0}(w)\).  Hence, we shall use the following unifying notation:
\begin{eqnarray}
\label{ells}
g_{m}(w) \equiv  
\prod_{\ell=0}^{n} (w-z_{\ell})^{s_{m\ell}},\quad m=1,\dots,d_0,
\end{eqnarray}
where \(s_{d_00}=0\) if \(z_0=0\) is not a root, and, for \(\ell>0\),
\(s_{d_0\ell}\) is the multiplicity of \(z_\ell\) as a root of 
\(g_{d_0}(w)\). For \(1\leq m<d_0\), \(s_{m\ell}\) is {\em either} 
the multiplicity of \(z_\ell\) as a root of \(g_m(w)\), {\em or} 
it vanishes if \(z_\ell\) is not a root of \(g_m\) to begin with.  
The projector  \(\pi\equiv \sum_{m=d_0+1}^d |m\rangle\langle m|\) 
keeps track of the vanishing entries on the main diagonal of \(D(w)\),
\((\mathds{1}-\pi) D(w)=D(w)(\mathds{1}-\pi)=D(w)\).  

\begin{lem}
\label{localbulkstates}
Let \(\D\) be the Smith normal form of \(\A\), and \(\rho_d(\pi)\) 
the unique BBL projector such that 
\(({\bf 1}-\rho_d(\pi))\D= \D({\bf 1}-\rho_d(\pi))=\D\). Then, 
\( \Ker\D=\W_\D \oplus \Range \rho_d(\pi), \) with
\begin{eqnarray}
\label{wdforyou}
\hspace*{-1cm} \W_\D
\equiv \Span \{
\Phi_{z_\ell,v}|m\rangle\, |\,m=1,\dots,d_0;\ \ell=1,\dots,n;\ v=1,\dots,s_{m\ell}\}.
\end{eqnarray}
\end{lem}  
\begin{proof}
Since \(\D\Psi=\sum_{m=1}^{d_0}|m\rangle \g_m\{\psi_{jm}\}_{j\in\Z}\),
a sequence is annihilated by \(\D\) if and only if it is of the form
\(\Psi=\sum_{m=d_0+1}^d|m\rangle\{\psi_{jm}\}_{j\in\Z}\),
or \(\Psi=\sum_{m=1}^{d_0} |m\rangle\{\psi_{jm}\}_{j\in\Z}\) and
\(\g_m\{\psi_{jm}\}_{j\in\Z}=0\) for all \(m\leq d_0\). 
According to Eq.\,\eqref{ells},
\( \g_m=\prod_{\ell=0}^{n} (\T-z_{\ell})^{s_{m\ell}}, \)
and so its kernel may be determined by inspection with the help of
Lemma\,\ref{scalarker}. 
\end{proof}

\(\Range \rho_d(\pi)\) is an uncountably infinite-dimensional 
space. Since we have not specified a topology on any space of sequences 
(we will only do so briefly in Appendix \ref{ibbtwhynot}), there is no easy way to 
describe a basis for it. Fortunately, for the purpose of investigating 
the kernel of \(A_N=\P_{L,R}\A|_{\V_{L,R}}\), we only need to characterize 
\(\F^{-1}\Pib_0\), where \(\Pib_0\) denotes the subspace 
of sequences in \(\Ker\rho_d(\pi)\) of {\em finite support}. The sequences
\[ \Delta_{i}|m\rangle=\{\delta_{ij}|m\rangle\}_{j\in\Z}
\,,\quad i\in\Z,\quad m=d_0+1,\dots,d,
\]
where \(\delta_{ij}\) denotes the Kronecker delta, form a basis of 
\(\Pib_0\). For convenience, we will call the space \(
\F^{-1}\W_\D\oplus\F^{-1}\Pib_0 \equiv {\rm Ker}_c\, \A \subset \Ker\A  
\) 
the {\em countable kernel} of \(\A\). Putting things together, we
obtain a complete and explicit description of \({\rm Ker}_c\,\A\).
We will write 
\[ \F^{-1}=\rho_d(F^{-1}(w)), \quad\text{with}\quad 
F^{-1}(w)\equiv \sum_{s=0}^{ \deg F^{-1}}w^s \hat{f}_s,
\quad \deg F^{-1}\in\N, \quad \hat{f}_s\in\M_d .\]

\begin{thm}
\label{thyglorywithf}
The countable kernel of  \(\A\) is spanned by the sequences  
\begin{align*}
\hspace{-1cm}
\F^{-1}\Delta_i|m\rangle&=
\sum_{s=0}^{\deg F^{-1}}\Delta_{i-s}\hat{f}_s|m\rangle, 
\quad  
m=d_0+1,\dots,d;\ i\in \Z,\\
\F^{-1}\Phi_{z_\ell,v}|m\rangle&=
\sum_{x=1}^{s_{m\ell}} \Phi_{z_\ell,x} [ F^{-1}_{s_{m\ell}} ( z_\ell)]_{x v}|m\rangle,
\, 
m=1,\dots,d_0;\,\ell=1,\dots,n;\,v=1,\dots,s_{m\ell}.
\end{align*}
\end{thm}
\begin{proof}
Since $ {\rm Ker}_c\, \A = \F^{-1} \W_\D \oplus \F^{-1} \Pib_0$, 
the claim follows from computing the action of  
\(\F^{-1}=\rho_d(F^{-1}(w))=\sum_{s=0}^{ \deg F^{-1}}\T^s \hat{f}_s\)
on the explicit bases of \(\Pib_0\) (directly) and \(\W_D\), as defined in 
Eq. (\ref{wdforyou}), with the
assistance of Lemma \ref{sdamaps}. 
\end{proof}

If \(d_0=d\), then \(\A\) is regular and \(\Ker \A\) is 
finite-dimensional. In this case, \(\Ker \A\) may actually be determined 
{\em without} explicitly computing the Smith factorization of 
\(\A\).

\begin{thm}
\label{smithbypass}
\label{smithbypass2}
Let $\A=\rho_d(A(w,w^{-1}))$ be a regular BBL transformation, 
and let 
\begin{equation}
\det A(w,w^{-1})=cw^{dp}\prod_{m=1}^dg_m(w)\equiv c \,w^{dp}\prod_{\ell=0}^n(w-z_\ell)^{s_\ell},
\label{eq:smith}
\end{equation}
where \(s_\ell=\sum_{m=1}^d s_{m\ell}\) is the multiplicity of 
\(z_\ell\) as a root of \(\det A= 0\),
and \(s_0=0\) if \(z_0=0\) is not a root. 
Then, the following properties hold: 
\begin{enumerate}
\item 
\( \Ker\A =\bigoplus_{\ell=1}^n \Ker A_{s_\ell}  (z_\ell) . \)
\item \( \dim \Ker A(z_\ell,z_\ell^{-1})\leq \dim \Ker A_{s_\ell} (z_\ell)=s_\ell, \) 
for any \(\ell=1,\dots,n\). 
\end{enumerate}
Furthermore, if the inequality in {\it (ii)} saturates, then 
\(\Ker A_{s_\ell} (z_\ell) =\Span \{\Phi_{z_\ell,1}|u_s\rangle\}_{s=1}^{s_\ell}\),
where \(\{|u_s\rangle\}_{s=1}^{s_\ell}\) is a basis of \(\Ker A(z_\ell,z_\ell^{-1})\).
\end{thm}

\begin{proof}
{\em (i)} We know an explicit basis of \(\Ker\A\) from 
Theorem \ref{thyglorywithf}. By grouping together the basis 
vectors associated to each nonzero root \(z_\ell\), one obtains 
the decomposition 
\[
\Ker \A=
\bigoplus_{\ell=1}^n{\cal W}_{z_\ell}, \quad {\cal W}_{z_\ell}= 
\F^{-1}{\rm Span} \{\Phi_{z_\ell,v}|m\rangle\,\big|\, v=1,\dots,s_{m\ell}; m=1,\dots,d\},
\]
so that  
\(\dim{\cal W}_{z_\ell}=\sum_{m=1}^d s_{m\ell}=s_\ell\). 
Moreover, by Lemma \ref{sdamaps}(iii),
\(
{\cal W}_{z_{\ell}}\subset\F^{-1}{\cal T}_{z_\ell,s_\ell}={\cal T}_{z_\ell,s_\ell}.
\) 
By construction, \(\A\) has no kernel vectors in \({\cal T}_{z_\ell,s_\ell}\)
other than the ones in \({\cal W}_{z_\ell}\). Hence,
\( {\cal W}_{z_\ell}=\Ker \A\cap {\cal T}_{z_\ell,s_\ell}=\Ker A_{s_\ell}(z_\ell).\)
We conclude that \(\dim\Ker A_{s_\ell} (z_\ell) =s_\ell\), the 
multiplicity of \(z_\ell\neq 0\) as a root of \(\det A\).

{\em (ii)} Let ${\mathbf D} =\rho_d (D(w))$ denote the Smith normal form of ${\mathbf A}$.
Because \( g_1(w)|g_2(w)|\dots|g_d(w), \)
a non-zero root \(z_\ell\) of \(\det A(w,w^{-1})\) 
will appear for the first time in one of the \(g_m(w)\), 
say \(g_{m_\ell}(w)\), and reappear with equal or greater 
multiplicity in every \(g_{m}(w)\) with \(d\geq m>m_\ell\). 
In particular, \(g_m(z_\ell)\neq 0\) if \(m<m_\ell\),
and vanishes otherwise. Hence, 
\(
\dim\Ker D(z_\ell)=\dim\Ker A(z_\ell,z_\ell^{-1})=d-m_\ell+1. \)
The number \(m_\ell\) is as small as possible whenever
\(z_\ell\) is a root with multiplicity one of {\em each} of 
the \(g_{m_\ell}(w),\dots,g_{d}(w)\), in which 
case \(m_\ell=d-s_\ell+1\). This shows that
\(\dim\Ker A(z_\ell,z_\ell^{-1})\leq s_\ell=\dim\Ker A_{s_\ell} (z_\ell ) \).  

Suppose next that \(\dim\Ker A(z_\ell,z_\ell^{-1})=s_\ell\),
and let \(\{|u_s\rangle \}_{s=1}^{s_\ell}\) denote a basis
of this space. Then, 
\(
\A\Phi_{z_\ell,1}|u_s\rangle
=\Phi_{z_\ell,1}A(z_\ell,z_\ell^{-1})|u_s\rangle=0,\) \(s=1,\dots,s_\ell
\). Since \(\Ker  A_{s_\ell} (z_\ell) \)
is \(s_\ell\)-dimensional, it follows that these linearly independent 
sequences span the space.
\end{proof}

In Ref.\,\cite{abc}, for simplicity, we considered problems 
such that the condition 
\(\dim\Ker A(z_\ell,z_\ell^{-1})=s_\ell\) was met.
The meaning of this assumption can be understood in terms of 
the above Theorem. If that is not the case, then it becomes necessary 
to work with the larger but still finite-dimensional matrix 
\( A_{s_\ell} (z_\ell)\) instead.

\subsubsection{Finite-support solutions.}
\label{extend}

According to Theorem \ref{softy}, if the principal coefficients 
of the BBT matrix \(A_N=\P_{L,R}\A|_{\V_{L,R}}\) fail to be invertible 
(which can happen even for $d =1$), the bulk solution space
\(\mathcal{M}_{L,R}=\Ker P_B A_N\) contains, but need {\em not} be 
contained in  \(\P_{L,R} \Ker \A\). The solutions of the bulk equation 
that are  {\em not} in \(\P_{L,R} \Ker \A\) were referred to as emergent 
before. We aim to establish a structural characterization of \(\mathcal{M}_{L,R}\). 
Our strategy will be to characterize the spaces  \(\mathcal{M}_{L,\infty}\) and 
\(\mathcal{M}_{-\infty, R}\) first [recall Definition \ref{mlr}], and then proceed to establish
their relationship to \(\mathcal{M}_{L,R}\).   

\begin{thm}
\label{thm:semi}
For $\A$ regular, let \(\sigma \equiv d\tau-\dim\Ker \A\), and 
\(\overline{L} \equiv L+\sigma-1\), \(\overline{R}\equiv R-\sigma +1\).
Then, there exist spaces 
\(\mathcal{F}_L^{-} \subseteq \V_{L,\overline{L}}\) and 
\(\mathcal{F}_R^{+} \subseteq \V_{\overline{R},R}\) such that 
\[
\mathcal{M}_{L,\infty} = \P_{L,\infty}\Ker\A \oplus \mathcal{F}_L^{-},
\quad\mbox{and}\quad 
\mathcal{M}_{-\infty,R} = \P_{-\infty,R}\Ker\A \oplus \mathcal{F}_R^{+}. \]
\end{thm}

\begin{proof}
It suffices to prove the claim for \(\mathcal{M}_{L,\infty}\);
the reasoning is the same for \(\mathcal{M}_{-\infty,R}\).
We first establish that \(\mathcal{M}_{L,\infty}\)
is finite-dimensional.
Since \(\dim \P_{L,R}\mathcal{M}_{L,\infty}\) is 
finite-dimensional if \(R-L\) is finite, the question becomes
whether there are nonzero sequences in \(\mathcal{M}_{L,\infty}\) 
annihilated by \(\P_{L,R}\). The answer is negative 
provided that \(R-L\geq \tau\). 
By contradiction, suppose that 
$\Psi \in \mathcal{M}_{L,\infty}$ satisfies \(\P_{L-p',\infty}\A\Psi=0=\P_{L,R}\Psi\). Then, 
\( \A\Psi=\P_{-\infty,L-p'-1}\A\Psi=
\sum_{r=p}^qa_r\T^r\P_{-\infty,L+(r-p')-1}\Psi=0, 
\)
because \(L+(r-p')-1\leq L+(q-p')-1\leq L+(q'-p')\leq R\). 
By construction, the translated sequences \(\{\Psi_n=\T^n\Psi\}_{n\in\mathds{Z}}\) 
are linearly independent and satisfy \(\A\Psi_n=0\), implying
that \(\dim\Ker \A=\infty\). By Theorem \ref{smithbypass2},
this contradicts the regularity of \(\A\). 
Hence,  it must be \(\P_{L,R}\Psi\neq 0\), and so 
\(\dim\mathcal{M}_{L,\infty}=\dim\P_{L,R}\mathcal{M}_{L,\infty}<\infty,\) 
as was to be shown. 
In conjunction with Theorem \ref{lemker}, this implies that 
\( \dim\mathcal{M}_{L,\infty}=\dim\P_{L,R}\mathcal{M}_{L,\infty}\leq
    \dim\mathcal{M}_{L,R}=d\tau. \)

While $\mathcal{M}_{L,\infty}$ is finite-dimensional just like 
\(\mathcal{M}_{L,R}\), what is special about this space is that
it is an invariant subspace of a translation-like transformation,
the {\it unilateral shift} $\P_{L,\infty}\T$. To see that
this is the case, notice that 
\(
\P_{L-p',\infty}\A \P_{L,\infty}\T\mathcal{M}_{L,\infty}=
\T\P_{L+1-p',\infty}\A\P_{L+1,\infty}\mathcal{M}_{L,\infty}=0, 
\)
because of the nesting property 
\( \P_{L+1,\infty}\mathcal{M}_{L,\infty}\subseteq\mathcal{M}_{L+1,\infty}. \)
Hence, \( \P_{L,\infty}\T\mathcal{M}_{L,\infty}\subseteq\mathcal{M}_{L,\infty}, \)
as was to be shown. Since $\mathcal{M}_{L,\infty}$ is finite-dimensional,
it can be decomposed into the direct sum of generalized eigenspaces of 
\(\P_{L,\infty}\T|_{\mathcal{M}_{L,\infty}}\). Accordingly, let us write 
\( \mathcal{M}_{L,\infty} \equiv \mathcal{N}\oplus\mathcal{F}_L^-, \)
where \(\P_{L,\infty}\T|_{\mathcal{F}_L^{-}}\) is nilpotent 
and $\P_{L,\infty}\T|_\mathcal{N}$ is invertible. 

The space $\mathcal{F}_L^{-}$, that is, the generalized kernel 
$\P_{L,R}\T|_{\mathcal{M}_{L,\infty}}$, is a subspace of 
$\V_{L,\overline{L}}$, defined in Eq. (\ref{Vlr}). The reason is that, if 
\(\Psi\in\mathcal{F}_L^{-}\),
then there exists a smallest positive integer \(\kappa\), its rank, 
such that \((\P_{L,\infty}\T)^\kappa\Psi=0\). The rank \(\kappa\) obeys 
\( \kappa \leq \dim \mathcal{F}_L^{-}=
\dim \mathcal{M}_{L,\infty}-\dim \P_{L,\infty}\Ker\A=
d\tau-\dim\Ker \A =\sigma.  \) 
As a consequence, if $\Psi \in \mathcal{F}_L^{-}$, then 
$\P_{L,\infty}\T^\sigma\Psi=(\P_{L,\infty}\T)^{\sigma}\Psi = 0,$ 
and so \( \Psi \in \V_{L,\overline{L}}\), with \(\overline{L}=L+\sigma-1\).  
The space \(\mathcal{N}\), that is, the direct sum of all
the generalized eigenspaces of $\P_{L,R}\T|_{\mathcal{M}_{L,\infty}}$
associated to non-zero eigenvalues, coincides with 
\(\P_{L,\infty}\Ker\A\). To see that this is the case, let 
\( \Psi_n = (\P_{L,\infty}\T)^{-n}\Psi\in\mathcal{N}, \)
for any \(n \in \N \), and $\Psi\in\mathcal{N}$. 
Since  \(
\T^n\Psi_n \in \T^n\mathcal{M}_{L,\infty} = \mathcal{M}_{L-n,\infty},\)
and
\( \P_{L,\infty}\T^n\Psi_n = (\P_{L,\infty}\T)^n\Psi_n = \Psi, \) 
we conclude that $\Psi \in \P_{L,\infty}\mathcal{M}_{L-n,\infty}$ for any 
$n$. It follows that  
$
\Psi \in \P_{L,\infty}\mathcal{M}_{-\infty,\infty} = \P_{L,\infty}\Ker\A.
$ 
The opposite inclusion $\P_{L,\infty}\Ker\A \subseteq \mathcal{N}$ 
holds because 
\(  \P_{L,\infty}\T\P_{L,\infty}\Ker\A=\P_{L,\infty}\P_{L-1,\infty}\T\Ker\A=
\P_{L,\infty}\Ker\A. \)
\end{proof}

The desired relationship between the spaces \(\mathcal{M}_{L,\infty}\) and 
\(\mathcal{M}_{-\infty, R}\) and the bulk solution space \(\mathcal{M}_{L,R}\) is 
contained in the following: 

\begin{lem}
\label{lem:directsum}
\(\mathcal{M}_{L,R} = 
\Span \big(\P_{L,R}\mathcal{M}_{-\infty,R}\cup\P_{L,R}\mathcal{M}_{L,\infty}\big).
\)
\end{lem}
\begin{proof}
The inclusion
\(
\Span \big(\P_{L,R}\,\mathcal{M}_{L,\infty}\cup\P_{L,R}\,\mathcal{M}_{-\infty,R}\big) 
\subseteq \mathcal{M}_{L,R}
\) follows from the nesting property, Eq.\eqref{nesting}. 
The task is to prove the opposite inclusion. Let us first 
show that 
\begin{equation}
 \P_{L,R}\,\Ker\A = \P_{L,R}\,\mathcal{M}_{-\infty,R}  \cap \P_{L,R}\,\mathcal{M}_{L,\infty}.
\label{equality}
\end{equation}
Again, because of nesting, 
\( \P_{L,R}\,\Ker\A\equiv\P_{L,R}\mathcal{M}_{-\infty,\infty} \subseteq 
\P_{L,R}\,\mathcal{M}_{-\infty,R} \cap
\P_{L,R}\, \mathcal{M}_{L,\infty}. \)
Take now an arbitrary element 
\( \{|\chi_j\rangle\}_{j=L}^R \in
\P_{L,R}\mathcal{M}_{L,\infty} \cap \P_{L,R}\mathcal{M}_{-\infty,R}\subseteq\mathcal{M}_{L,R}.
\)
By definition, there exist sequences
$\Psi_1\equiv \{|\psi_{1j}\rangle\}_{j=L}^{\infty}\in \mathcal{M}_{L,\infty}$ 
and 
$\Psi_2  \equiv \{|\psi_{2j}\rangle\}_{j=-\infty}^{R}\in \mathcal{M}_{-\infty,R}$, 
such that 
\( \P_{L,R}\Psi_1 =  \{|\chi_j\rangle\}_{j=L}^R = \P_{L,R}\Psi_2. \)
Let $\Psi$ denote the unique sequence with 
\(\P_{-\infty,R}\Psi=\Psi_1\) and \(\P_{L,\infty}\Psi=\Psi_2\). Then,
\( \A\Psi=(\P_{-\infty,R-q'}+\P_{L-p',\infty}-\P_{L-p',R-q'})\A\Psi=0,
\)
confirming that  
\(\{|\chi_j\rangle\}_{j=L}^R=\P_{L,R}\Psi\in \P_{L,R}\Ker \A\) and 
proving the equality in Eq. (\ref{equality}). 
It then follows that
\begin{align*}
\dim \Span\big(\P_{L,R}\mathcal{M}_{-\infty,R}\cup\P_{L,R}\mathcal{M}_{L,\infty}\big)=
\dim\mathcal{M}_{-\infty,R}+ \dim\mathcal{M}_{L,\infty}- \dim\Ker\A.
\end{align*}
Since the right hand-side is independent of \(L,R\) (recall that, 
by construction, $\mathcal{M}_{L,\infty} = \T^{L'-L} \mathcal{M}_{L',\infty}$, 
and similarly for $\mathcal{M}_{-\infty,R}$), 
this dimension is thus {\em independent of} \(L,R\).

The next step is to show that there exists an integer \(\infty> R_0\geq R\) 
such that\footnote{At an intuitive level, this result is very appealing: 
it implies that a measurement on sites \(L\) to \(R\) cannot tell 
\(\mathcal{M}_{L,\widetilde{R}}\) apart from \(\mathcal{M}_{L,\infty}\)
if \(\widetilde{R}\) is large enough.}  
\begin{equation}
\label{R0}
\P_{L,R}\mathcal{M}_{L,\widetilde{R}}=\P_{L,R}\mathcal{M}_{L,\infty},\quad
\forall \widetilde{R}\geq R_0.
\end{equation}
We first show that
\( \mathcal{M}_{L,\infty} = 
\bigcap_{n\geq 0} \left(\mathcal{M}_{L,R+n}\oplus \V_{R+n+1,\infty}\right), 
\)
or, more explicitly, that  
\begin{equation*}
\Ker\P_{L-p',\infty}\A|_{\V_{L,\infty}} = 
\bigcap_{n\geq 0} \Big(\Ker \P_{L-p',R+n-q'}\A|_{\V_{L,R+n}} \oplus \V_{R+n+1,\infty}\Big).
\end{equation*}
The equality holds because a sequence \(\Psi\) belongs to 
either space if and only if it satisfies the set of equations
\(\sum_{r=p}^{q}a_{r}|\psi_{j+r}\rangle=0,\) for all \(j\geq L-p'+p. \)
As a consequence, 
\begin{eqnarray}
\label{intersec2}
\P_{L,R}\mathcal{M}_{L,\infty} = \bigcap_{n\geq 0} \P_{L,R}\mathcal{M}_{L,R+n} .
\end{eqnarray}
Because of the nesting property of bulk solution spaces, the 
spaces on the right hand-side of Eq.\,\eqref{intersec2} satisfy 
\(
\P_{L,R}\mathcal{M}_{L,R+n_2}= 
\P_{L,R}\P_{L,R+n_1}\mathcal{M}_{L,R+n_2}\subseteq
\P_{L,R}\mathcal{M}_{L,R+n_1}$, for all \(n_2\ge n_1.
\) 
Therefore, if \(\delta(n)\equiv\dim\P_{L,R}\mathcal{M}_{L,R+n}\), \(n\geq 0\)
(in particular, \(\delta(0)=d(q'-p')\)), then
\(\delta(n_2)\leq \delta(n_1)\). Since \(\delta\) is a non-decreasing 
function bounded below, there exists \(n_0\) such that 
\(\delta(n)=\delta(n_0)\equiv\delta_0\) 
for all \(n\geq n_0\). Then Eq.\,(\ref{intersec2}) implies 
that  
$\P_{L,R}\mathcal{M}_{L,\infty} = \P_{L,R}\mathcal{M}_{L,R+n}$
if $n\ge n_0$, thus establishing Eq. (\ref{R0}), with $R_0\equiv R+n_0$.

Thanks to the special properties of \(R_0\), now we can 
prove a special instance of our main claim, namely, the equality 
\begin{eqnarray}
\label{toshownow}
\mathcal{M}_{L,R_0}=
\Span \big(\P_{L,R_0}\mathcal{M}_{-\infty,R_0}\cup\P_{L,R_0}\mathcal{M}_{L,\infty}\big).
\end{eqnarray}
By definition of \(R_0\), if  
\(\Psi\in\mathcal{M}_{L,R_0}\), then  \(\P_{L,R}\Psi\in\P_{L,R}\mathcal{M}_{L,\infty}\).
Thus, there exists \(\Upsilon\in \mathcal{M}_{L,\infty}\) 
such that \(\P_{L,R}\Upsilon=\P_{L,R}\Psi\). Let 
\(\Psi\equiv \Psi_1+\Psi_2,\) with 
\( \Psi_1\equiv \P_{L,R_0}\Upsilon\) and \(\Psi_2=\Psi-\P_{L,R_0}\Upsilon.\)
Since \(\Psi_1\in \P_{L,R_0}\mathcal{M}_{L,\infty}\) by construction,
it only remains to show that 
\(\Psi_2\in \P_{L,R_0}\mathcal{M}_{-\infty,R_0}\).
By nesting, \(\P_{L,R_0}\Upsilon\in \mathcal{M}_{L,R_0}\) and 
so \(\Psi_2\in\mathcal{M}_{L,R_0}\). In particular,
\(\P_{-\infty,L-1}\Psi_2=0\). Hence,
\(
\P_{-\infty,R_0-q'}\A\Psi_2=\P_{-\infty,L-p'-1}\A\Psi_2
=\sum_{r=q}^{p}a_r\T^r\P_{L,L+r-p'-1}\Psi_2=0,
\)
because  \(
\P_{L,R}\Psi_2=\P_{L,R}\Psi-\P_{L,R}\Upsilon=0 \)
(by the way \(\Upsilon\) was chosen), and 
\(L+r-p'-1< R\). It follows that \(\Psi_2\in\mathcal{M}_{-\infty,R_0}\),
and, since \(\P_{L,R_0}\Psi_2=\Psi_2\), 
that \(\Psi_2\in \P_{L,R_0}\mathcal{M}_{-\infty,R_0}\).
This concludes the proof of Eq.\,\eqref{toshownow}.

Our main claim follows from this special instance, because
the dimension of the span in question is independent of \(L,R\), 
and $\dim\mathcal{M}_{L,R}=d\tau$ by Theorem \ref{lemker}.
\end{proof}

Putting these results together, we are now in a position to give the 
anticipated structural characterization of the bulk solution space: 

\begin{thm}
\label{thm:fintosemi}
If \(A_N=\P_{L,R}\A|_{\V_{L,R}}\) is regular, and
$N\geq 2\sigma+\tau$, with $\sigma = d\tau-\dim\Ker \A$, then
\(
\mathcal{M}_{L,R} = \P_{L,R}\Ker\A \oplus \mathcal{F}_L^{-} \oplus \mathcal{F}_R^{+}.  \)
\end{thm}

\begin{proof} 
Because of the lower bound on $N$, 
$\P_{L,R}\mathcal{F}_L^{-} =\mathcal{F}_L^{-} $ and 
$\P_{L,R}\mathcal{F}_R^{+} = \mathcal{F}_R^{+}$. Using 
Lemma \,\ref{lem:directsum} and the direct sum decompositions 
in Thm.\,\ref{thm:semi}, the only additional result required 
to prove the theorem is
$\{\P_{L,R}\Ker\A \oplus \mathcal{F}_L^{-}\} \cap \mathcal{F}_R^{+} = \{0\}.$
By contradiction, assume that there exists some non-zero vector $\Psi$ 
in this intersection. Then $\Psi \in \mathcal{F}_L^{+}$ implies 
$\P_{L,\bar{R}-1}\Psi=0$, which can be split into two conditions 
$\P_{L,\bar{L}}\Psi=0,\ \P_{\bar{L}+1,\bar{R}-1}\Psi=0$. Since 
$\Psi \in \P_{L,R}\Ker\A \oplus \mathcal{F}_L^{-}$, we may also 
express $\Psi$ as $\Psi = \Psi_1+\Psi_2$, where $\Psi_1 \in \P_{L,R}\Ker\A$ 
and $\Psi_2 \in \mathcal{F}_L^{-}$. Now since $\P_{\bar{L}+1,\bar{R}-1}\Psi_2=0$, 
therefore the second of the two equations imply $\P_{\bar{L}+1,\bar{R}-1}\Psi_1=0$. 
The lower bound on $N$ implies that $\bar{R}-1-(\bar{L}+1)+1>\tau$, 
so that $\Psi_1=0$. Then the 
first equation leads to $\Psi_2=0$, so that $\Psi=\Psi_1+\Psi_2=0$, 
which is a contradiction. 
\end{proof}

According to the above theorem, \(\mathcal{M}_{L,R}\) consists
of three qualitatively distinct contributions: bulk solutions 
associated to \(\P_{L,R}\Ker \A\), bulk solutions ``localized''
near \(L\), and bulk solutions ``localized'' near \(R\), with the latter
two types being the emergent solutions of finite support. 
Remarkably, this characterization brings together all three 
relevant length scales of our eigensystem problem: the size \(N\equiv R-L+1\) 
of the BBT matrix of interest, the distance \(\tau\equiv q'-p'\), which 
is the lower bound on \(N\) associated to a non-trivial bulk, 
and the support bound \(\sigma \equiv d\tau-\dim\Ker \A\) for solutions
of finite support.
Formally, we may obtain a basis of the bulk solution space \(\mathcal{M}_{L,R}\) 
by combining bases of the three subspaces
$\P_{L,R}\Ker\A,\ \mathcal{F}_{L}^{-}$, and $\mathcal{F}_{R}^{+}$.
We have extensively discussed bases for \(\P_{L,R}\Ker\A\), recall in 
particular Theorem \ref{smithbypass} for \(\A\) regular. 
Let us now define square matrices
\[
K^- \equiv \P_{L-p',\overline{L}-p'}\A|_{\V_{L,\overline{L}}}\, 
\quad \mbox{and}\quad  
K^+ \equiv \P_{\overline{R}-q',R-q'}\A|_{\V_{\overline{R},R}}. 
\]
We may then further give the following characterization:

\begin{coro}
\label{cor:kmatrix}
If $N\geq\tau+2\sigma$, then 
$\mathcal{F}_L^{-} = \Ker K^-$ and $\mathcal{F}_R^{+} = \Ker K^+$.
\end{coro}
\begin{proof}
We will show only that $\mathcal{F}_L^{-} = \Ker K^-$;
the other equality follows analogously.
By Theorem \ref{thm:fintosemi}, $\mathcal{F}_L^{-} \subseteq \V_{L,\overline{L}}$
and 
\(\mathcal{F}_R^{+}\cap\V_{L,\overline{L}}
=\{0\}=
\P_{L,R}\Ker \A\cap\V_{L,\overline{L}}
\) (remember that \(\A\) is regular). Hence, 
$\mathcal{F}_L^{-}=\Ker \P_{L-p',R-q'}\A|_{\V_{L,\overline{L}}}$, and 
\(
\P_{L-p',R-q'}\sum_{r=p}^{q}a_r\T^r|_{\V_{L,\overline{L}}}
=\sum_{r=p}^{q}\P_{L-p',\overline{L}-r}a_r\T^r|_{\V_{L,\overline{L}}}=
\P_{L-p',\overline{L}-p'}\A|_{\V_{L,\overline{L}}}\equiv
K^-,
\)
because \(\overline{L}-r\leq\overline{L}-p\leq\overline{L}-p' \). 
\end{proof}

\begin{remark}
\label{rem:singular}
If $\A$ is singular, the  kernels of $K^+$ and $K^-$ 
are still contained in the bulk solution space. Together with 
Theorem \ref{softy}, this observation implies that 
\[
\mathcal{M}_{L,R} \supseteq 
 \Span\big(\P_{L,R}\Ker\A  \cup (\mathcal{F}_L^{-} \oplus \mathcal{F}_R^{+})\big),
\quad\quad \A\ \mbox{singular}, 
\]
\end{remark}
\noindent 
however, additional solutions to the bulk equation may  exist on a case-by-case basis.

\subsection{An exact result for the boundary equation}

As we already remarked in Sec. \ref{seckerinclusion}, the boundary equation is 
generally associated to an unstructured matrix, the boundary matrix [Definition \ref{boundmat}]. 
While its solution thus relies in general on numerical methods, 
there is one exact result that follows from our work so far. We isolate it here:

\begin{thm}
\label{corboundmat}
If \(A_N\) is regular, then \(\dim \Ker P_BA_N=\dim \Range P_\partial\),
and the boundary matrix of the corner-modified BBT matrix \(C=A_N+W\)
is square independently of the corner modification \(W\). 
\end{thm}
\begin{proof}
The boundary matrix is the matrix of the compatibility map
\(B=P_\partial (A_N+ W)|_{\Ker P_BA_N}\), see Section\,\ref{seckerinclusion}. 
Hence, the number of rows of $B$ is determined by 
\( \dim \Range P_\partial= d\tau\). The number of columns is determined by 
the dimension of $\text{Ker} P_BA_N$, which, by Lemma \ref{lemker}, is also 
also $d\tau$ if \(A_N\) is regular. 
\end{proof}

\subsection{Multiplication of corner-modified banded block-Toeplitz matrices }

In general, the product of two BBT matrices is not a 
BBT matrix (see Ref.\,\cite{fardad09} for a lucid discussion of this point). 
However, the product of two corner-modified BBT matrices is again a corner
modified BBT matrix, provided they are large enough.  As a consequence, the 
problem of determining the generalized kernel of a corner-modified BBT matrix \(C\) 
is equivalent to that of determining the kernel of \(C^\kappa\) in the same class. 

\begin{thm} 
\label{superduper}
Let $C_i=A_{N,i}+W_i$, $i=1,2$, denote corner-modified BBT 
matrices, with \(A_{N,i}=\P_{L,R}\A_i|_{\V_{L,R}}\) of bandwidth 
\((p_i,q_i)\) and {\(2(N-1)>q_1-p_1+q_2-p_2 \).} Then,
the product \(  C_1C_2= A_N+W \)
is a corner-modified BBT matrix of bandwidth \((p_1+p_2,q_1+q_2)\),
with \(A_N= \P_{L,R} \A_1\A_2|_{\V_{L,R}}\) and $W$ a corner modification for
this bandwidth.
\end{thm}

\begin{proof}
We proceed by first showing that 
\begin{eqnarray}
\label{times1}
P_B C_1C_2=P_B(A_{N,1}+W_1)(A_{N,2}+W_2)=P_BA_{N,1}A_{N,2},
\end{eqnarray}
where \(P_B\) is the bulk projector for bandwidth
\((p_1+p_2,q_1+q_2)\). In particular, \(P_BP_{B,1}=P_B\), where 
\(P_{B,1}\) is the bulk projectors for bandwidth \((p_1,q_1)\). 
Hence,   
$P_{B}W_1(A_{N,2}+W_2)=P_BP_{B,1}W_1(A_{N,2}+W_2)=0$ because
\(P_{B,1}W_1=0\). Morever, since
\[
P_BA_{N,1}P_{\partial,2}|\psi\rangle=P_B\sum_{j=1}^{-p_2'-p_1'}|j\rangle|\phi_j\rangle+
P_B\sum_{N-q_2'-q_1'+1}^N|j\rangle|\phi_j\rangle=0, \quad \forall 
|\psi\rangle, \]
one concludes that \(P_BA_{N,1}P_{\partial,2}=0\) 
and \(P_BA_{N,1}P_{\partial,2}W_2=P_BA_{N,1}W_2=0\) as well. 
The next step is to show that
\begin{eqnarray}\label{times2}
P_{B,1}A_{N,1}A_{N,2}=P_{B,1}\A_1\A_2|_{\V_{L,R}}.
\end{eqnarray}
By definition, \(P_{B,1}A_{N,1}A_{N,2}=P_{B,1}\P_{L,R}\A_1\pn\A_2|_{\V_{L,R}}=
P_{B,1} \A_1\P_{L,R}\A_2|_{\V_{L,R}}\). However,   
\(
P_{B,1} \A_1\A_2|_{\V_{L,R}}-P_{B,1} \A_1\P_{L,R}\A_2|_{\V_{L,R}}=
P_{B,1}\A_1({\bf 1}-\P_{L,R}) \A_2|_{\V_{L,R}}=0.
\)
The reason is that, on the one hand, the sequences in the range of  
\(({\bf 1}-\P_{L,R})\A_2|_{\V_{L,R}}\) necessarily vanish on sites 
\(L\) to \(R\). On the other hand, acting with \(\A_1\) on any 
such sequence produces a sequence that necessarily vanishes on 
sites \(L-p_1'\) to \(R-q_1'\). 
Such sequences are annihilated by \(P_{B,1}\). Combining Eqs.\,\eqref{times1} 
and \eqref{times2}, we conclude that 
\(
P_B \A_1\A_2|_{\V_{L,R}}=P_BA_{N,1}A_{N,2}=P_BC_1C_2,
\)
and so the bulk of the product \(C_1C_2\) coincides with the bulk of the
BBT matrix \(A_N=\P_{L,R}\A_1\A_2|_{\V_{L,R}}\).  
\end{proof}

\noindent 
The following corollary follows immediately, by repeated application of the above result:

\begin{coro}
\label{corproduct}
Let \(C=A_N+W\) denote a corner-modified BBT transformation of bandwidth
\((p,q)\) and \(A_N=\P_{L,R}\A|_{\V_{L,R}}\). Then,
 $C^\kappa= A_N^{(\kappa)} +W_\kappa$ is a corner-modified BBT matrix as 
as long as as long as \(2(R-L)>\kappa(q-p)\). The bandwidth 
of $A_N^{(\kappa)}\equiv \P_{L,R}\A^\kappa|_{\V_{L,R}}$ is 
\((\kappa p,\kappa q)\), and \(W_\kappa\) a corner modification for this 
bandwidth. 
\end{coro}

\section{Algorithms} 
\label{sec:algorithms}

\subsection{The kernel algorithm}
\label{seckeralgo}

Having provided a rigorous foundation to our approach of kernel determination by 
projectors, we now outline an algorithmic procedure for constructing a basis of the 
kernel $\mathcal{M}_{1,N}$ of a corner-modified BBT matrix, $C= A_N+W$.
The input $C$ is given in the form of the matrix coefficients $\{a_r,\ p\le r \le q\}$ 
of $A(w,w^{-1})$, and the block-entries $[W]_{ij} \equiv \langle i | W | j \rangle$ of 
its corner-modification $W$. Throughout this section,
we assume that 
$A(w,w^{-1})$ is {\em regular}.  
We divide the algorithm in two steps: 
I. Solving the bulk equation; II. Constructing the boundary matrix and solving the 
associated kernel equation. 

\subsubsection*{I. Solution of the bulk equation.}
\label{secbulkeq}

We proved in Lemma \ref{softy} that if the principal coefficients are invertible, then 
$\mathcal{M}_{1,N}= \P_{1,N}\text{Ker}\A$. For non-invertible principal coefficients, 
$\mathcal{M}_{1,N}$ is the direct sum of   $\P_{1,N}\Ker\A$, 
$\mathcal{F}_1^{-}$, and $\mathcal{F}_N^{+}$ [Theorem \ref{thm:fintosemi}]. 
We now construct a basis of each of these subspaces.
According to Theorem \ref{smithbypass}, 
$\Ker\A = \bigoplus_{\ell=1}^{n}\Ker\A \cap \mathcal{T}_{z_\ell,s_\ell},$
where $\{z_\ell\}_{\ell=1}^{n}$ are {\em non-zero} roots of the characteristic 
equation \( \det A(w, w^{-1}) =0 \), 
and $\{s_\ell\}_{\ell=1}^{n}$ their multiplicities [see Eq. (\ref{eq:smith})].
Each of the subspaces $\mathcal{T}_{z_\ell,s_\ell}$ is invariant under the action 
of $\A$, and $\A|_{\mathcal{T}_{z_\ell,s_\ell}}$ has representation $A_{s_\ell}(z_\ell)$ 
in the canonical basis of Lemma \ref{sdamaps}. Therefore, a block-vector 
$|u\rangle = \big[|u_1\rangle \dots |u_{s_\ell}\rangle\big]^{\rm T}$, 
$\{|u_v\rangle \in \mathds{C}^d\}_{v=1}^{s_\ell}$, belonging to $\text{Ker}A_{s_\ell}(z_\ell)$, 
represents the sequence 
\begin{eqnarray*}
\Psi_{\ell s} \equiv \sum_{v=1}^{s_\ell}\Phi_{z_\ell,v}|u_v\rangle \in \Ker\A \cap \mathcal{T}_{z_\ell,s_\ell}
\end{eqnarray*}
Its projection on $\vone$, namely, 
\[ 
|\psi_{\ell s}\rangle \equiv \P_{1,N} \Psi_{\ell s} = \sum_{v=1}^{s_\ell}|z_\ell,v\rangle|u_v\rangle 
\in {\bm P}_{1,N} \Ker {\bm A}, 
 \quad |z_\ell,v\rangle = \sum_{j=1}^{N}j^{(v-1)}z^{j-v+1}|j\rangle,\]
is a solution of the bulk equation. A basis $\mathcal{B}_\text{ext}^{(\ell)} \equiv \{|\psi_{\ell s}\rangle\}$ of $\P_{1,N}\Ker\A$, corresponding to root $z_\ell$, may thus be inferred from a basis $\{|u_{\ell s}\rangle\}_{s=1}^{s_\ell}$ of ${\rm Ker}\,A_{s_\ell}(z_\ell)$, where 
\begin{eqnarray}
|u_{\ell s}\rangle = \big[|u_{\ell s 1}\rangle \dots |u_{\ell s s_\ell}\rangle\big]^{\rm T},\quad \{|u_{\ell s v}\rangle \in \mathds{C}^d\}_{v=1}^{s_\ell} \ \  \forall \ell,s.
\label{eq:b1}
\end{eqnarray}
A basis of $\P_{1,N}\Ker\A$ is given by 
$\mathcal{B}_\text{ext}\equiv \bigcup_{\ell=1}^{n}\mathcal{B}_\text{ext}^{(\ell)},$ 
with this basis being stored in the form of vectors 
$\big\{|u_{\ell s}\rangle\,|\,s=1,\dots,s_\ell;\ \ell=1,\dots, n\}$,  
with $|\mathcal{B}_\text{ext}|=\dim \P_{1,N}\Ker\A  = \sum_{\ell=1}^n s_\ell$.

If the principal coefficients of $\A$ are not invertible, one needs to additionally obtain bases of $\mathcal{F}_1^{-}$ and $\mathcal{F}_N^{+}$. According to Corollary\,\ref{cor:kmatrix}, the kernel of $\mathcal{F}_1^{-}$ ($\mathcal{F}_N^{+}$) coincides with ${\rm Ker}\,K^-$ (${\rm Ker}\,K^+$). In the standard bases $\{|j\rangle\}_{j=1}^{\sigma}$ and $\{|j\}_{j=1-p}^{\sigma-p}$ of the subspaces $\V_{1,\sigma}$ and $\V_{1-p,\sigma-p}$, respectively, $K^-$ is a block-matrix of size $d\sigma \times d\sigma$, with block-entries 
\begin{eqnarray*}
\label{matkminus}
[K^{-}]_{jj'} = \left\{ \begin{array}{cl} a_{j-j'+p'} & \text{if} \ \ j' \le j\le j'+\tau \\ 0 & \text{otherwise} 
\end{array} \right. , \quad 1\le j, j' \le \sigma.
\end{eqnarray*}
Every vector 
$|u\rangle = \big[|u_1\rangle \dots |u_{\sigma}\rangle\big]^{\rm T}$, 
with $\{|u_j\rangle\in \mathds{C}^d\}_{j=1}^{\sigma}$, 
in $\text{Ker}K^-$ provides a corresponding (finite-support) solution of the bulk equation, namely, 
\begin{eqnarray*}
|\psi\rangle = \sum_{j=1}^{\sigma}|j\rangle|u_j\rangle \in \mathcal{F}_1^{-}.
\end{eqnarray*}
Accordingly, a basis 
$\mathcal{B}^- \equiv \{|\psi_s^{-}\rangle\}_{s=1}^{s_{-}}$ 
of $\mathcal{F}_1^{-}$ may be stored as the basis $\{|u_s^{-}\rangle\}_{s=1}^{s_-}$ 
of $\Ker K^{-}$, where $s_- \equiv {\rm dim}(\Ker K^-)$ and 
\begin{eqnarray}
|u_s^{-}\rangle = 
\big[|u_{s 1}^-\rangle \dots |u_{s \sigma}^-\rangle\big]^{\rm T},\quad \{|u_{s j}^-\rangle 
\in \mathds{C}^d\}_{j=1}^{\sigma} \ \forall s.
\label{eq:b2}
\end{eqnarray}
Similarly, a basis of $\mathcal{F}_N^{+}$ can be obtained from a basis of $K^+$, with entries
\begin{eqnarray*}
\label{matkplus}
[K^{+}]_{jj'} = \left\{ \begin{array}{cl} a_{j-j'+q'} & \text{if} \ \ j \le j' \le j+\tau \\ 0 & \text{otherwise} \end{array} \right. , \quad 1\le j, j' \le \sigma.
\end{eqnarray*}
In this case, each $|u\rangle = \big[|u_1\rangle \dots |u_{\sigma}\rangle\big]^{\rm T}$, 
$\{|u_j\rangle\in \mathds{C}^d\}_{j=1}^{\sigma}$, in $\text{Ker} K^+$ represents the 
(finite-support) solution of the bulk equation given by 
\begin{eqnarray*}
|\psi\rangle = \sum_{j=1}^{\sigma}|N-\sigma+j\rangle|u_j\rangle \in \mathcal{F}_N^{+}.
\end{eqnarray*}
Then a basis 
$\mathcal{B}^+ \equiv \{|\psi_s^{-}\rangle\}_{s=1}^{s_{+}}$ of 
$\mathcal{F}_N^{+}$ is stored as the basis 
$\{|u_s^{+}\rangle\}_{s=1}^{s_+}$ of $\Ker K^{+}$, where 
$s_+ \equiv {\rm dim}(\Ker K^+)$ and 
\begin{eqnarray}
|u_s^{+}\rangle = 
\big[|u_{s 1}^+\rangle \dots |u_{s \sigma}^+\rangle\big]^{\rm T},\quad \{|u_{s j}^+\rangle\in \mathds{C}^d\}_{j=1}^{\sigma} \ \forall s.
\label{eq:b3}
\end{eqnarray}
If the principal coefficients of $\A$ are invertible, then both $\mathcal{B}^{-}$ and $\mathcal{B}^{+}$ are empty.  The procedure is summarized in box \ref{algo:bulk}.

\begin{figure}[t]
\begin{myalg}{Kernel algorithm I: Solution of the bulk equation}
\label{algo:bulk}
\hspace{3mm} {\sc Input:} Matrix Laurent polynomial $A(w,w^{-1})$ and $N$.
\begin{enumerate}[leftmargin=0.5in, rightmargin=0.2in]

\item Find all non-zero roots of $\det A(z,z^{-1})=0$. Let these be denoted by $\{z_\ell\}_{\ell=1}^{n}$.

\item \label{step1.3} For each root, construct the matrix $A_{s_\ell}(z_\ell)$.

\item \label{step1.4} Find a basis $\{|u_{\ell s}\rangle\}_{s=1}^{s_\ell}$ of the kernel of $A_{s_\ell}(z_\ell)$.  

\item If the principal coefficients are not invertible, construct the matrices $K^-$ and $K^+$ with block-entries given in Eqs.\,(\ref{matkminus})-(\ref{matkplus}).

\item Compute bases $\{|u_s^{-}\rangle\}_{s=1}^{s_-}$ and $\{|u_s^{+}\rangle\}_{s=1}^{s_+}$ of ${\rm Ker}\,K^-$ and ${\rm Ker}\,K^+$ respectively.

\end{enumerate}
\hspace{3mm} {\sc Output:} $\{z_\ell\}_{\ell=1}^{n}$, $\big\{\{|u_{\ell s}\rangle\}_{s=1}^{s_\ell}\big\}_{\ell=1}^{n}$, $\{|u_s^-\rangle\}_{s=1}^{s_-}$, $\{|u_s^+\rangle\}_{s=1}^{s_+}$.
\medskip
\end{myalg}
\end{figure}

\subsubsection*{II. Construction of the boundary matrix and solution of the boundary equation.} 
The union of the three bases, $\mathcal{B}\equiv \mathcal{B}_\text{ext} \cup  \mathcal{B}^{-}\cup \mathcal{B}^{+}$, 
provides a basis of $\mathcal{M}_{1,N}$, the entire solution space of the bulk equation. As long as the matrix Laurent polynomial $A(w,w^{-1})$ is regular, the number of basis vectors in $\mathcal{B}$ is $d\tau$ [Lemma \ref{lemker}]. Let $\{ |\psi_s\rangle, s = 1,\dots,d\tau\}$ be the basis vectors in $\mathcal{B}$, where each $|\psi_s\rangle$ is expressible as $|\psi_s\rangle = \sum_{j=1}^{N}|j\rangle|\psi_{sj}\rangle$. We next construct a matrix representation of the boundary map $B \equiv P_\partial C\big|_{\mathcal{M}_{1,N}}$, using $\mathcal{B}$ as the basis of $\mathcal{M}_{1,N}$. The entries of this matrix are then 
\begin{equation}
\label{bmatentries}
{[B]_{bs}= \langle b| B |\psi_s\rangle = \langle b| (A_N+W) |\psi_s\rangle } = 
\sum_{r=\max(p,-b+1)}^{\min(q,N-b)}a_r|\psi_{s \,b+r}\rangle + \sum_{j=1}^{N}[W]_{b\,j}|\psi_{s\, j}\rangle,
\end{equation}
where, as noted, $s=1,\dots,d\tau$ and $b$ takes the values given in Eq.\,(\ref{rangeb}).
Note that if the corner modification $W$ is {\em symmetrical}, we may further observe that 
\begin{eqnarray*}
\hspace*{-1cm}
{[B]_{bs} =  \langle b | (A_N+W) |\psi_s\rangle = 
 \langle b | A_N|\psi_s\rangle + \langle b |(P_\partial W Q_\partial) Q_\partial |\psi_s\rangle, \quad \forall s.}
\end{eqnarray*}
{The entries of $B$ may then be computed more efficiently by using 
\begin{equation}
\label{bmatentriessym}
\langle b |(P_\partial W Q_\partial) Q_\partial |\psi_s\rangle = \left(\sum_{j=1}^{-p} + \sum_{j=N-q+1}^{N} \right)[W]_{bj}|\psi_{sj}\rangle , 
\end{equation}
which makes it clear that the number of terms in each sum is independent of $N$.}

The final step is to construct a basis of ${\text{Ker}}\, C$ from the boundary matrix. We compute such a 
basis in the form $\{\bm{\alpha}_k\}_{k=1}^{n_C}$,  where each basis vector is expressed as
\begin{eqnarray*}
\bm{\alpha}_k = \big[\alpha_1 \dots \alpha_{d\tau}\big]^{\rm T},\ \{\alpha_{ks} \in \mathds{C}\}_{s=1}^{d\tau}.
\end{eqnarray*}
The entries $\alpha_{ks}$ of each $\bm{\alpha}_k$ provide the coefficients of the bulk solutions $|\psi_s\rangle$ in the linear combination, that forms a kernel vector of $C$. For instance, 
$ |\epsilon_k\rangle = \sum_{s=1}^{d\tau}\alpha_{ks} |\psi_s\rangle $
is a vector in $\text{Ker} \, C$. Then $\{|\epsilon_k\rangle\}_{k=1}^{n_C}$ 
form a basis of $\Ker C$. The block-entries of each vector $|\epsilon_k\rangle$ can be easily calculated, since 
$\langle j|\epsilon_k\rangle = \sum_{s=1}^{d\tau}\alpha_{ks} |\psi_{sj}\rangle,$ for all $k,$
and $|\psi_{sj}\rangle$ are the entries of $|\psi_s\rangle$, which are stored in a compact form as output of Algorithm \ref{algo:bulk}. The output of this part of the algorithm is the basis $\{|\epsilon_k\rangle\}_{k=1}^{n_C}$ of $\Ker C$, again in the compact form of $\{\bm{\alpha}_k\}_{k=1}^{n_C}$, along with the output of the previous algorithm. The procedure is summarized in box \ref{algo:bound}.

\begin{figure}[t]
\begin{myalg}{Kernel algorithm II: Solution of \\the boundary equation}
\label{algo:bound}
\hspace{3mm} {\sc Input:}  $A(w,w^{-1})$, $W$ and output of Algorithm \ref{algo:bulk}.
\begin{enumerate}[leftmargin=0.5in, rightmargin=0.2in]
\item Construct the boundary matrix $B$ using the formula given in 
Eq.\,(\ref{bmatentries}) for non-symmetrical and 
Eq.\,(\ref{bmatentriessym}) for symmetrical corner modifications, respectively.

\item Find a basis $\{\bm{\alpha}_k\}_{k=1}^{n_C}$ of the kernel of $B$.
\end{enumerate}
\hspace{3mm} {\sc Output:} $\{\bm{\alpha}_k\}_{k=1}^{n_C}$ and output of Algorithm \ref{algo:bulk}.
\medskip
\end{myalg}
\end{figure}

\subsection{The multiplication algorithm}
\label{seckernel}

We will now describe an efficient algorithm for multiplying two corner-modified BBT transformations 
$C_{i}=A_{N,i}+W_i$, $i=1,2$, where each $A_{N,i}$ is a BBT transformation of bandwidth $(p_i,q_i)$, and $W_{i}$ are the corresponding corner modifications. The input for the algorithm are the associated matrix Laurent polynomials $A_i(w,w^{-1})$.
Let $C=C_1 C_2$ which, by Theorem\,\ref{superduper},  
is also a corner-modified BBT transformation. Our task is to calculate efficiently the matrix Laurent polynomial $A(w,w^{-1})$ associated to the BBT matrix $A_N$, and the entries of the corner-modification $W$, that satisfy $C=A_N+W$. Theorem\,\ref{superduper} implies that $A(w,w^{-1})=A_1 (w,w^{-1}) A_2 (w,w^{-1})$, the calculation of which involves finding $(p_1+p_2,q_1+q_2)$ matrix coefficients that are easily obtained from the coefficients of $A_i(w,w^{-1})$, $i=1,2$. The remaining task is computation of the entries of $W$. 
Corollary \ref{corproduct} leads to the expression
\[W = P_\partial(A_{N,1}A_{N,2}-A_N) +  (A_{N,1}P_{\partial, 2})(P_{\partial, 2}W_2) 
+  (P_{\partial, 1}W_1)A_{N,2}
+ (P_{\partial, 1} W_1P_{\partial, 2})(P_{\partial, 2}W_2), \]
and thereby to the formula 
\begin{multline}
\label{cmentries}
{[W]_{bj} = \sum_{j'=\max(1,b-q_1)}^{\min(N,b-p_1)} 
a_{1,b-j'} a_{2, j'-j} - a_{b-j}  + \sum_{j'=\max(1,b-q_1)}^{\min(N,b-p_1)} a_{1,b-j'} [W_2]_{j'j} } \\ 
{+ \sum_{j'=\max(1,p_2+j)}^{\min(N,q_2+j)}[W_1]_{bj'}a_{2, j'-j}+ 
\left(\sum_{j'=1}^{-p_2}+\sum_{j'=N-q_2+1}^{N}\right)[W_1]_{bj'}.[W_2]_{j'j}.  }
\end{multline}
The algorithm outputs $A(w,w^{-1})$, along with all the entries of $W$, which completely 
describes the product transformation $C$. The procedure is summarized in box \ref{algo:mult}. 

\begin{figure}[t]
\begin{myalg}{Multiplication algorithm}
\label{algo:mult}

\hspace{3mm} {\sc Input:} $\{ A_i(w,w^{-1}), W_i \}_{i=1}^{2}$.
\begin{enumerate}[leftmargin=0.5in, rightmargin=0.2in]

\item Compute coefficients of the matrix Laurent polynomial 
$A(w,w^{-1})=A_1(w,w^{-1}) A_2 (w,w^{-1})$.

\item Compute all entries of $W$ using Eq.\,(\ref{cmentries}).
\end{enumerate}
\hspace{3mm} {\sc Output:} $A(w,w^{-1}), W$.
\medskip
\end{myalg}
\end{figure}

\subsection{The eigensystem algorithm}
\label{algo}

Given a corner-modified BBT transformation $C=A_N+W$, the goal 
is to obtain its spectrum and a basis of the corresponding generalized eigenvectors. Again, we divide this algorithm in two parts. Part I computes the spectrum of $C$ and finds corresponding eigenvectors. If the latter span $\vone$, then there exist no generalized eigenvectors of higher rank, and the problem is solved. If not, 
part II finds generalized eigenvectors corresponding to all the eigenvalues, already obtained in the first step. 

\vspace*{1mm}

\noindent 
{\em I. Eigenvalues and eigenvector determination.}
This is a particular instance of an appropriate root-finding algorithm on $\mathds{C}$, where eigenvalues of $C$ are the desired roots. Conventionally, the eigenvalue problem of a linear operator $M$ is viewed as a root-finding problem, since eigenvalues of $M$ are roots of its characteristic equation, $\det (M-\epsilon)=0$. The algorithm we propose does not seek roots of the characteristic equation of $C$, but instead of a {\em function whose roots coincide with those of the characteristic equation}. This function is the determinant of the boundary matrix of the corner-modified BBT transformation $C-\epsilon = (A_N-\epsilon) + W$, whose kernel is the eigenspace of $C$ corresponding to eigenvalue $\epsilon$. 
If $B(\epsilon)$ denotes the boundary matrix of $C-\epsilon$, then the problem of finding the spectrum of $C$ is equivalent to that of finding roots of the equation $\det B(\epsilon)=0$. The kernel algorithm described in Sec.\,\ref{seckeralgo} can be implemented to compute 
$B(\epsilon)$ for each value of $\epsilon$. Whenever $\epsilon$ is an eigenvalue, the kernel algorithm also provides the corresponding eigenvectors. Typically, $\det B(\epsilon)$ is a continuous complex-valued function of $\epsilon$, a feature that can be leveraged in implementing an appropriate root-finding algorithm of choice. 

\begin{figure}[t]
\begin{myalg}{Eigensystem algorithm I: Solution of \\ eigenvalue problem} 
\label{algo:eig}
\hspace{3mm} {\sc Input:} $A(w,w^{-1}), W$.
\begin{enumerate}[leftmargin=0.5in, rightmargin=0.2in]

\item \label{stepone} Find all values of $\epsilon$ for which $A(w,w^{-1})-\epsilon$ is singular. 

\item \label{4.2} If $W$ is symmetrical, output all values in step\,(\ref{stepone}) as eigenvalues. If not, then compute $\det(C-\epsilon)$ for each, and output those that have zero determinant.

\item For each eigenvalue found in step\,(\ref{4.2}), find and output a basis of $\Ker (C-\epsilon)$.
\item Choose a seed value of $\epsilon$ different from any of the values found in step (\ref{stepone}). 

\item \label{loop} Find $B(\epsilon)$ using the kernel algorithm, with $A(w,w^{-1})-\epsilon$ and $W$ as inputs. 

\item If $\det B(\epsilon)=0$, then output $\epsilon$ as an eigenvalue. Output a basis of $\Ker (C-\epsilon)$ from $B(\epsilon)$ as described in the kernel algorithm. This is a basis of the eigenspace of $C$ corresponding to eigenvalue $\epsilon$.

\item Choose a new value of $\epsilon$ as dictated by the relevant root-finding algorithm. Go back to step (\ref{loop}).  

\end{enumerate}
\hspace{3mm} {\sc Output:} All eigenvalues of $C$ and bases of corresponding eigenspaces.
\medskip
\end{myalg}
\end{figure}

The kernel of $C-\epsilon$ coincides with $B(\epsilon)$ provided that the associated matrix Laurent polynomial is regular, which is the case generically. If there exist some values of $\epsilon$ for which $A(w,w^{-1})-\epsilon$ is singular, then whether or not those are part of the spectrum may be found by computing $\det (C-\epsilon)$ directly\footnote{Note that if the corner-modification is symmetrical, 
then the boundary equation is trivially satisfied by the bulk solutions 
localized sufficiently away from either boundary, implying that these values of $\epsilon$ are always part of the spectrum in such cases.}. Remarkably, such a singular behavior can occur  only at a few isolated value of $\epsilon$ \cite{matrix_polynomials}. 
The procedure is summarized in box \ref{algo:eig}.

{\em II. Generalized eigenvectors determination}.
In this case, the eigenvectors obtained in part I do not span the entire space $\vone$. 
For each eigenvalue $\epsilon$, let $(C-\epsilon)^\kappa = A_N^{(\epsilon,\kappa)}+W_{\epsilon,\kappa}$ define the relevant matrix Laurent polynomial $A^{(\epsilon,\kappa)}(w,w^{-1})$ and the corner-modification $W_{\epsilon,\kappa}$. Starting from $\kappa=2$, we first compute $A^{(\epsilon,\kappa)}(w,w^{-1})$ and $W_{\epsilon,\kappa}$ using the multiplication algorithm. Next, we construct its boundary matrix $B(\epsilon,\kappa)$ using the kernel algorithm. The dimension of  $\Ker (C-\epsilon)^\kappa$ is the same as that of $\Ker B(\epsilon,\kappa)$. If $\dim\{\Ker B(\epsilon,\kappa)\} >\dim\{\Ker B(\epsilon,\kappa-1)\}$,
then there exists at least one generalized eigenvector of $C$ of rank $\kappa$. In this case, we compute $A^{(\epsilon,\kappa+1)}(w,w^{-1})$ and $W_{\epsilon,\kappa+1}$. We repeat this process, until we find a value $\kappa_{\max}$ for which $\dim\{\Ker B(\epsilon,\kappa_{\max})\} = \dim\{\Ker B(\epsilon,\kappa_{\max}+1)\}$. This indicates that there are no generalized eigenvectors of $C$ of rank greater than $\kappa_{\max}$ corresponding to eigenvalue $\epsilon$. Then, from the boundary matrix of $(C-\epsilon)^{\kappa_{\max}}$, we compute a basis of the generalized eigenspace corresponding to $\epsilon$. This process is repeated for every eigenvalue $\epsilon$ to obtain bases of the corresponding eigenspaces.  A basis of $\vone$ 
is obtained by combining all these bases.

In the non-generic case where $A(w,w^{-1})-\epsilon$ is singular for some eigenvalue $\epsilon$, we can still use the multiplication algorithm to find $(C-\epsilon)^\kappa$, but the corresponding kernel and its dimension are found using some conventional algorithm. The algorithm summarized box \ref{algo:geneig} is {\em provably complete} for those eigenvalues for which $A(w,w^{-1})-\epsilon$ is regular.

\begin{figure}[th]
\begin{myalg}{Eigensystem algorithm II: Solution of \\ generalized eigenvalue problem}
\label{algo:geneig}
\hspace{3mm} {\sc Input:} All eigenvalues and the dimensions of corresponding eigenspaces.
\begin{enumerate}[leftmargin=0.5in, rightmargin=0.2in]

\item \label{start} Choose any of the eigenvalues, call it $\epsilon$.

\item \label{loop2} Set $\kappa = 2$.

\item \label{product} Find $A^{(\epsilon,\kappa)}(w,w^{-1})$ and $W_{\epsilon,\kappa}$ using the multiplication algorithm.

\item Construct the corresponding boundary matrix $B(\epsilon,\kappa)$.

\item If $\dim\Ker B(\epsilon,\kappa)>\dim\Ker B(\epsilon,\kappa-1)$, then increment $\kappa$ by one and go back to step (\ref{product}). If not, set $\kappa_{\max} = \kappa$.

\item Find a basis of the kernel of $(C-\epsilon)^{\kappa_{\max}}$ from the boundary matrix, as described in the Kernel algorithm. Choose a new eigenvalue and go back to step (\ref{loop2}).

\end{enumerate}
\hspace{3mm} {\sc Output:} Bases for all generalized eigenspaces of $C$.
\medskip
\end{myalg}
\end{figure}

\subsection{Efficiency considerations}
\label{secbettercorner}

It is important to ensure that both the number of steps 
(time complexity) and the memory space (space complexity) required by our  eigensystem 
algorithm scale favorably with the size $N\gg1$ of the corner-modified BBT matrix $C=A_N+W$ of interest. 

\subsubsection*{The kernel algorithm.}
The first part of the kernel algorithm concerning the solution of the bulk equation 
does not make any reference to the size of $A_N$. Specifically, 
we store the basis vectors in 
$\mathcal{B}^+,\mathcal{B}^-$ and $\mathcal{B}_\text{ext}$ 
in the form of vectors 
$\{|u_s^{+}\rangle\}_{s=1}^{s_+}, \{|u_s^{-}\rangle\}_{s=1}^{s_-}$ 
and 
$\{\{|u_{\ell,s} \rangle\}_{s=1}^{s_\ell}\}_{\ell=1}^{n}$, 
respectively, along with the roots 
$\{z_\ell, \ 1\le \ell \le n\}$. 
The same is true about the solution of the boundary equation, since the obtained basis 
of $\mathcal{M}_{1,N}$ is outputted in the form of vectors 
$\{|\alpha_k\rangle\}_{k=1}^{n_C}$. 
Thus, it is the second step, involving the construction 
of the boundary matrix, that determines the time and space complexity 
of the algorithm. If $W$ is not symmetrical, then computation of each entry of $B$ according to Eq.\,(\ref{bmatentries}) involves a summation of $N$ terms. Therefore, in the worst case, the algorithm requires $\mathcal{O}(N)$ time steps for complete kernel determination. However, in the important case where $W$ is symmetrical, the summation is only over $2d\tau$ terms, which is {\em independent of $N$}. In these cases, the time-complexity is $\mathcal{O}(1)$. Note that the storage units required by the algorithm scale as $\mathcal{O}(N)$ 
in the general case, because of the entries of the corner modification that need to be stored. The auxiliary space required is only $\mathcal{O}(1)$. For symmetrical corner modifications, the space required to store $W$ is $\mathcal{O}(1)$, which is also the space complexity of the kernel algorithm. 

\subsubsection*{The multiplication algorithm.}
Calculation of the matrix coefficients of the product matrix Laurent polynomial is a trivial task from the point of view of complexity. Also, according to Eq.\,(\ref{cmentries}), computing each entry of the resulting corner modification involves summations that do not grow with $N$. In the 
general case, the number of entries of the corner modification scales linearly with $N$, therefore the time and space complexities of multiplication algorithm are $\mathcal{O}(N)$. If the given transformations have both symmetrical corner modifications, then the resulting corner modification is also symmetrical. In these cases, the number of non-trivial entries of the resulting corner modification does not scale up with $N$, implying that both time and space complexities are $\mathcal{O}(1)$.

\section{Applications}
\label{sec:applications}

\subsection{An Ansatz for the eigenvectors of a corner-modified block-Toeplitz matrix}

Based on the analysis in Sec.\,\ref{secexact}, we may formulate an 
exact eigenvalue-dependent Ansatz for the eigenvectors 
of a given corner-modified BBT transformation. An
Ansatz of similar form, catering to some special circumstances, 
was introduced in $\cite{abc}$. 

Any eigenvector of a corner-modified BBT transformation $J$, 
corresponding to eigenvalue $\epsilon$, is a kernel vector of 
the transformation $C \equiv J-\epsilon$, which is also a 
corner-modified BBT transformation. 
Thanks to Eq.\,({\ref{boundarymap}), any kernel vector of $C$
satisfies its bulk equation, that is, the kernel equation for
\(P_BC=P_BA_N=P_B\P_{1,N}\A|_{\V_N}\), where we assume henceforth that 
$L=1$, $R=N$.
Further, if the principal 
coefficients of $C$ (that is, \(A_N\)) are invertible, then by 
Lemma \ref{softy} the solution space $\mathcal{M}_{1,N}$ of the 
bulk equation for $C$ is identical to $\P_{1,N}\Ker \A$. 
Otherwise, by Theorem \ref{thm:fintosemi},  
the solution space is 
$\mathcal{M}_{1,N}=\P_{1,N}\Ker\A \oplus \mathcal{F}_N^{+} 
\oplus \mathcal{F}_1^{-}$. In any case, as long as the BBL 
transformation \(\A\) associated to $C$ is {\em regular},  
$\mathcal{M}_{1,N}$ is $d\tau$-dimensional and its basis may 
be obtained by the union of the bases of the constituent 
subspaces. It follows that any kernel vector of $C=J-\epsilon$ may be expressed as
a linear combination 
\begin{equation}
\label{ansatz}
|\epsilon\rangle = 
\sum_{\ell=1}^{n}\sum_{s=1}^{s_\ell}\alpha_{\ell s}|\psi_{\ell s}\rangle
+\sum_{s=1}^{s_+}\alpha^{+}_{s}|\psi^{+}_{s}\rangle  
+\sum_{s=1}^{s_-}\alpha^{-}_{s}|\psi^{-}_{s}\rangle\ ,
\end{equation}
where the complex coefficients 
\(\alpha_{\ell s},\alpha^{+}_{s},\alpha^{-}_{s} \in \mathds{C}\) are parameters to be determined 
and 
\begin{eqnarray*}
|\psi_{\ell s}\rangle = 
\sum_{v=1}^{s_\ell}|z_\ell,v\rangle|u_{\ell s v}\rangle \in \P_{1,N}\Ker\A\ ,\label{ansinf}\\
|\psi^{+}_{s}\rangle = 
\sum_{j=1}^{d\tau}|N-d\tau+j\rangle|u^{+}_{s j}\rangle\in \mathcal{F}_N^{+}\ ,\\
|\psi^{-}_{s}\rangle = 
\sum_{j=1}^{d\tau}|j\rangle|u^{-}_{s j}\rangle \in \mathcal{F}_1^{-} \ ,
\end{eqnarray*} 
for basis vectors described in Eqs. (\ref{eq:b1})-(\ref{eq:b3}).
An Ansatz for generalized eigenvectors of rank $\kappa>1$ 
can be obtained similarly, since $C^\kappa = (J-\epsilon)^\kappa$ is 
also a corner modified BBT matrix, as shown in Corollary \ref{corproduct}.

\begin{remark}
Theorem \ref{thm:fintosemi} applies only to those cases where the 
matrix Laurent polynomial under consideration is regular. Therefore, 
the Ansatz in Eq.\,(\ref{ansatz}) is {\em provably complete} only for 
those corner-modified BBT matrices $A_N+W$, where the associated matrix 
Laurent polynomial $A(w,w^{-1})-\epsilon$ is {\em regular for every $\epsilon$}, 
which is usually the case. If $A(w,w^{-1})-\epsilon$ is singular for {\em some} 
$\epsilon$,  we know that 
$\P_{1,N}\Ker\A$ and $\mathcal{F}_1^{-}\oplus \mathcal{F}_N^{+}$ are 
subspaces of $\mathcal{M}_{1,N}$ [Remark \ref{rem:singular}].
However, they need {\em not} span the entire $\mathcal{M}_{1,N}$. Such cases are 
important but rare, and typically correspond to some exactly solvable limits. 
In these cases, $\epsilon$ is a highly degenerate eigenvalue, with $\mathcal{O}(N)$ 
eigenvectors of the form given in Theorem \ref{thyglorywithf}, that have 
{\em finite support in the bulk}. For example, in free-fermionic Hamiltonians as considered 
in \cite{abc,prb1}, {\em dispersionless (``flat'') energy bands} form for such eigenvalues. 
A summary of our results on the structural characterization of the
bulk solution space $\mathcal{M}_{1,N}$ is given in Table \ref{table1}. 
\end{remark}

\begin{table}
\setlength\extrarowheight{10pt}
\centering
\begin{tabular} {| c | c | c |}
\hline
\backslashbox{$G(w)$}{$a_{p'},a_{q'}$} & Invertible & Non-invertible \\ \hline
Regular & $\mathcal{M}_{1,N}=\P_{1,N}\Ker\A$ & $\mathcal{M}_{1,N}=\P_{1,N}\Ker\A \oplus \mathcal{F}_1^{-} \oplus \mathcal{F}_N^{+}$ \\[10pt] \hline
Singular & ---
& $\mathcal{M}_{1,N}\supseteq\Span\big(\P_{1,N}\Ker\A  \cup (\mathcal{F}_1^{-} \oplus \mathcal{F}_N^{+})\big)$\\[10pt] \hline
\end{tabular}
\caption{\label{table1}
Structural characterization of the bulk solution space, 
depending on the invertibility of the principal coefficients and regularity of the corresponding matrix polynomial.}
\end{table}

If the principal coefficients of the associated matrix Laurent polynomial 
$G(w)$ are invertible, as considered in \cite{abc},
both the second and third terms in the Ansatz of 
Eq.\,(\ref{ansatz}) vanish. Irrespective of the invertibility of the 
principal coefficients, a simplification in the first term occurs if 
$s_\ell = \dim\,{\rm Ker}\,A(z_\ell,z_\ell^{-1})$ for some $\ell$, where $s_\ell$ 
is the algebraic multiplicity $z_\ell$ as a root of \(\det A(w,w^{-1})\) (recall
that \(s_\ell=\dim \Ker A_{s_\ell}(z_\ell)\)).  In these cases, 
Lemma \ref{smithbypass} 
implies that each of the 
$|\psi_{\ell,s}\rangle$ in Eq.\,(\ref{ansatz}) has the simple form
\begin{eqnarray*}
|\psi_{\ell , s}\rangle = |z_\ell, 1\rangle |u_{\ell s 1}\rangle \quad \Rightarrow \quad
\big|\langle j|\langle m |\psi_{\ell s}\rangle\big| = 
\big|\langle m|u_{\ell s 1}\rangle z_\ell^j\big| \propto |z_\ell^j| ,
\label{expsols}
\end{eqnarray*}
with no contributions from terms $ |z_\ell, v \rangle |u_{\ell s 1}\rangle$ with $v>1$. We call vectors
of the above form {\it exponential solutions}, because their amplitude as a function of 
the lattice coordinate varies exponentially with $j$, for any \(m=1,\dots,d\). If $z$ lies on the unit circle, 
these solutions correspond to plane waves, with amplitude that is independent of $j$.

The condition 
$s_\ell=\dim \Ker A(z_\ell,z_\ell^{-1})$ is satisfied under generic situations by all roots 
$z_\ell$. In those special situations where $s_\ell > \dim\Ker A(z_\ell,z_\ell^{-1})$ for some $\ell$, one 
must allow for the possibility of $|\psi_{\ell,s}\rangle$ in Eq. \,(\ref{ansatz}) to describe what we refer to as {\it power-law solutions}, whose amplitude varies with $j$ as 
\begin{eqnarray*}
\big|\langle j|\langle m |\psi_{\ell s}\rangle\big| = 
\big|\sum_{v=1}^{s_\ell}\langle m|u_{\ell s v}\rangle j^{(v-1)}z_\ell^{j-v+1}\big| \propto |j^{s_\ell-1}z_\ell^{j-v+1}| ,
\quad \forall m .
\end{eqnarray*}
The Ansatz presented in Ref.\,\cite{abc} excludes power-law solutions by assuming that 
$s_\ell = \dim\Ker A(z_\ell,z_\ell^{-1})$ for every root $z_\ell$ of \(\det A(z,z^{-1})=0\). 

If the principal coefficients of the matrix Laurent polynomial 
$A(w,w^{-1})$ are not invertible, the contributions to the Ansatz in Eq. (\ref{ansatz})
that belong to  $\mathcal{F}_N^{+}$ and 
$\mathcal{F}_1^{-}$ are {\it finite-support solutions}.  This refers to the fact that,  
for all $m$, their  amplitude
\begin{eqnarray*}
\big|\langle j|\langle m |\psi_{s}^-\rangle\big| = 
\left\{ \begin{array}{cl} \big|\langle m|u_{s j}^-\rangle\big| & {\rm if}\ \  1 \le j \le d\tau \\ 
0 & {\rm if}\ \ j>d\tau \end{array} \right. , 
\end{eqnarray*}
for $j>d\tau$ in the case of $\mathcal{F}_1^{-}$, and similarly for $j<N-d\tau$ 
in the case of $\mathcal{F}_N^{+}$, respectively. The support of these solutions clearly does not 
change with $N$.

\subsection{The open-boundary Majorana chain revisited}
\label{secmajorana}

The Majorana chain \cite{kitaev01,bernevigbook} is the simplest tight-binding model of 
a (quasi) one-dimensional $p$-wave topological superconductor. For open boundary
conditions, and in second-quantization, the many-body Hamiltonian for a chain 
of length $N$ reads 
\[ \widehat{H}_N=-\sum_{j=1}^{N}\mu\,c_{j}^{\dagger}c_{j}-
\sum_{j=1}^{N-1}\left(t \, c_{j}^{\dagger}c_{j+1}-\Delta
\, c_{j}^{\dagger}c_{j+1}^{\dagger}+ \text{h.c.} \right),
\]
where $c^\dagger_j (c_j)$ are fermionic creation (annihilation) operators for the $j$th lattice 
site, and the parameters $\mu,t,\Delta \in \mathds{R}$ denote chemical potential, hopping 
and pairing strengths, respectively. Since the many-body Hamiltonian is 
quadratic, it is well-known that it suffices to diagonalize the corresponding single-particle 
Hamiltonian in Nambu space \cite{blaizot}. 
Following the derivation in \cite{abc}, the latter, $H_N \in \M_{2N}$, 
is found to be 
\begin{eqnarray*}
\label{kitaevHam}
H_N =T\otimes h_1+\mathds{1}\otimes h_0+T^\dagger\otimes h_1^\dagger=
\P_{1,N}\H|_{\V_N}, 
\end{eqnarray*}
with 
\begin{eqnarray*}
\hspace*{-1cm}
\H= h_1\T + h_0 \mathds{1} + h^\dagger_1 \Tminus\quad \mbox{and}\quad
h_0=-\begin{bmatrix}
\mu & 0\\ 
0 & -\mu 
\end{bmatrix}, \ \
h_1=-
\begin{bmatrix}t & -\Delta\\
 \Delta & -t
\end{bmatrix}.
\label{Kit}
\end{eqnarray*}
When restricted to the Hilbert space, $\H$ is precisely the Hamiltonian of the infinite Majorana chain.   
$H_N$ and $\H$ correspond, respectively, to a corner-modified BBT matrix on 
$\mathds{C}^N \otimes \mathds{C}^2$, with $p=p'=-1$ and $q=q'=1$  
and a symmetrical corner modification, 
and the associated BBL transformation. The principal coefficients, $a_{-1}=h_1$ and
$a_1=h_1^\dagger$, are invertible (hence the associated matrix 
Laurent polynomial is regular) in the generic 
regime $|t| \ne |\Delta|$, with arbitrary $\mu$, which we considered in Ref. \cite{abc}.

Here, we will diagonalize $H_N$ in the parameter regime $t=\Delta$, for arbitrary 
values of $\mu$ and $t$, corresponding to {\em non-invertible} principal coefficients. In particular, this will 
include the special case where, additionally, the system is tuned at $\mu=0$, which is referred to as the ``sweet spot'' in parameter space \cite{bernevigbook,fulga}.
Since $H_N$ is Hermitian, its eigenvectors span the entire $2N$-dimensional single-particle space $\mathds{C}^N \otimes \mathds{C}^2$, thus there are no generalized eigenvectors of rank greater than one.
In the following, we implement the kernel algorithm analytically with input $H_N-\epsilon$ for an arbitrary real value of $\epsilon$, 
leading to a closed-form solution of the eigensystem of interest. 
 
\subsubsection*{$\bullet$ Kernel determination for generic $\epsilon$.}  
The matrix Laurent polynomial $H(z,z^{-1})-\epsilon = z h_1 + h_0-\epsilon + z^{-1}h_1^\dagger$ and $N$ are provided as inputs to the  Algorithm \ref{algo:bulk}. In particular, 
the evaluation of $H(z,z^{-1})$ on the unit circle, $z \equiv e^{i k}, k\in\mathds{R}$, yields the 
Hamiltonian in momentum space, 
$$H_k \equiv \begin{bmatrix}-\mu-2t\cos k & i \, 2t \sin k \\ -i \, 2t \sin k & \mu +2t\cos k\end{bmatrix}.$$

\begin{enumerate}
\item The characteristic equation for $H(z,z^{-1})$ is 
\begin{eqnarray*}
\label{kitaevdispersion}
(z+z^{-1})(2\mu t)+(\mu^2+4t^2-\epsilon^2)=0.
\end{eqnarray*}
This is indeed an analytic continuation of the standard dispersion relation, 
which for $t=\Delta$ simply reads 
$\epsilon = \pm\sqrt{\mu^2+4t^2+4\mu t\cos k}$. The above equation 
has two non-zero roots, that we denote by $\{z_1,z_2\}$. For fixed values of $\mu, t$ and $\epsilon$, these two roots may be expressed analytically as $z_1=\zeta, z_2=\zeta^{-1}$, 
where
\begin{eqnarray*}
\zeta^{\pm1} = \frac{\epsilon^2-\mu^2-4t^2 \pm \sqrt{(\mu^2+4t^2-\epsilon^2)^2-16\mu^2t^2}}{4\mu t}.
\end{eqnarray*}

\item In order to construct the matrices $(H-\epsilon)_{s_\ell}(z_\ell)$ [recall Eqs. (\ref{Azs})-(\ref{AzsBlock})], 
we need to know the number of {\it distinct} roots and their multiplicities. The two roots coincide if and only if $\zeta=\pm1$, which 
happens if $\epsilon$ assumes one of the following special values:
\begin{equation}
\epsilon \in {\cal S} \equiv \{ \pm (\mu + 2t),\pm (\mu - 2t)\} .
\label{set}
\end{equation}
In these cases, there is only one distinct root ($1$ or $-1$), with algebraic multiplicity two. The implementation of the kernel algorithm for  $\epsilon \in {\cal S}$ will be shown separately. 
For $\epsilon \not\in {\cal S}$, the two roots have multiplicity one each. Then, 
with $z=\zeta,\zeta^{-1}$,  
\[  (H-\epsilon)_1 (z) = H(z,z^{-1})-\epsilon=
\begin{bmatrix}
-t (z + z^{-1}) - (\mu+\epsilon)  & t(z  - z^{-1}) \\ 
-t(z  - z^{-1}) & t(z +z^{-1}) + (\mu-\epsilon) 
\end{bmatrix}.\]

\item The kernel of $(H-\epsilon)_1(z)$ is spanned by 
\begin{eqnarray*}
|u(z)\rangle = \begin{bmatrix} t(z-z^{-1}) \\ \epsilon + \mu +t(z+z^{-1})\end{bmatrix},\quad z=\zeta,\zeta^{-1}.
\end{eqnarray*}
Then, $|u_{1,1}\rangle = |u(\zeta)\rangle$ and $|u_{2,1}\rangle = |u(\zeta^{-1})\rangle$ span $\Ker 
(H-\epsilon)_1(\zeta)$ and $\Ker (H-\epsilon)_1 (\zeta^{-1})]$, respectively. We have thus obtained the solutions of the bulk equation arising from the infinite problem, which are spanned by $\{|\zeta,1\rangle|u_{1,1}\rangle,|\zeta^{-1},1\rangle|u_{2,1}\rangle\}$, in the notation of Eq. (\ref{ansatz}). 
These are extended, exponential solutions. In particular, if $|\zeta|=1$, they correspond to plane waves. Note that these are {\it not} the kernel vectors of $H_N-\epsilon$, since the boundary conditions are not yet imposed. 

\item 
Since $d\tau=2\cdot 2=4$ and $\dim\Ker(\H-\epsilon)=2$, we have 
\begin{equation}
\label{kitaevK}
K_\epsilon^{-}=
\begin{bmatrix} 
h_1^\dagger & h_0-\epsilon \\ 
0 & h_1^\dagger 
\end{bmatrix}, \quad
K_\epsilon^{+}=
\begin{bmatrix} 
 h_1 & 0\\
 h_0-\epsilon & h_1
\end{bmatrix}.
\end{equation}

\item For any $\epsilon \not \in {\cal S}$, 
$s_-=s_+=1$, and the two one-dimensional kernels are spanned by 
$\{|u_1^{-}\rangle\}$ and $\{|u_1^{+}\rangle\}$, respectively, where
$$  | u_1^{-}\rangle \equiv \begin{bmatrix} |-\rangle \\ 0
\end{bmatrix}, \quad 
|u_1^{+}\rangle \equiv \begin{bmatrix} 0 \\  |+\rangle  \end{bmatrix}, $$
\noindent 
and $|+\rangle \equiv |1\rangle + |2\rangle$, $|-\rangle \equiv |1\rangle - |2\rangle,$ $|1\rangle,|2\rangle \in \mathds{C}^2.$
Therefore, $|1\rangle|-\rangle$ and $|N\rangle|+\rangle$ are the localized solutions of the bulk equation, 
one on each edge of the chain.
\end{enumerate}

\noindent
$\bullet$ {\em Boundary equation}. Now we will construct the boundary matrix, and 
determine the kernel of $H_N-\epsilon$ using Algorithm \ref{algo:bound}.
\begin{enumerate}
\item From Eq.\,(\ref{bmatentriessym}), the boundary matrix is found to be
\begin{eqnarray*}
\hspace*{-2.3cm}
B(\epsilon) = \begin{bmatrix}
-\mu-\epsilon & 0	& 2t^2\zeta+t(\epsilon+\mu) & 2t^2\zeta^{-1}+t(\epsilon+\mu) \\
-\mu+\epsilon & 0 & -2t^2\zeta-t(\epsilon+\mu) & -2t^2\zeta^{-1}-t(\epsilon+\mu) \\
0 & -\mu-\epsilon & \zeta^{N+1}(-2t^2\zeta^{-1}-t(\epsilon+\mu)) & \zeta^{-N-1}(-2t^2\zeta-t(\epsilon+\mu)) \\
0 & \mu-\epsilon & \zeta^{N+1}(-2t^2\zeta^{-1}-t(\epsilon+\mu)) & \zeta^{-N-1}(-2t^2\zeta-t(\epsilon+\mu))
\end{bmatrix}.
\end{eqnarray*}

\item The determinant of $B(\epsilon)$ is
\[ \det B(\epsilon)=4\mu^2t^2\Big(z^{N+1}(2tz^{-1}+\epsilon+\mu)^2 - z^{-(N+1)}(2tz+\epsilon+\mu)^2\Big).\]
Therefore, $B(\epsilon)$ has a non-trivial kernel if either of the conditions 
\begin{eqnarray}
\label{kitaevcondition1}
2t\zeta+\epsilon+\mu=\pm\zeta^{(N+1)}(2t\zeta^{-1}+\epsilon+\mu)
\end{eqnarray}
is satisfied. When this happens, the kernel of $B(\epsilon)$ 
is one-dimensional and is spanned by $|\alpha\rangle = [0 \; 0 \; 1 \; \mp\zeta^{N+1}] ^{\rm T}$. 
In turn, this implies that, when Eq.\,(\ref{kitaevcondition1}) is satisfied,
the kernel of $H_N-\epsilon$  is spanned by the vector 
\[ |\epsilon\rangle = |\zeta,1\rangle|u_{1,1}\rangle \mp \zeta^{N+1}|\zeta^{-1},1\rangle|u_{2,1}\rangle. \]
\end{enumerate}

\subsubsection*{$\bullet$ Solution for $\epsilon \in {\cal S}$.} We illustrate the case 
$\epsilon=\mu+2t$, as the analysis is similar for the other values in Eq.\,(\ref{set}).
Compared to the previous case, the implementation of Algorithm \ref{algo:bulk} differs only in the steps (\ref{step1.3})-(\ref{step1.4}). Since $\zeta=\zeta^{-1}=1$ has multiplicity two, the only matrix to be constructed in step (\ref{step1.3}) is [recall again Eqs. (\ref{Azs})-(\ref{AzsBlock})]
\begin{eqnarray*}
\hspace*{-1.4cm}
(H-\epsilon)_2(1)= \begin{bmatrix} H(1,1)-\epsilon & H^{(1)}(1,1) \\ 0 & H(1,1)-\epsilon \end{bmatrix} = 
2\begin{bmatrix} -2t-\mu & 0 & 0 & t \\ 0 & 0 & -t & 0 \\
0 & 0 & -2t-\mu & 0 \\ 0 & 0 & 0 & 0 \end{bmatrix}.
\end{eqnarray*}
The kernel of this matrix, computed in step (\ref{step1.4}), is spanned by $\{|u_{1,1}\rangle, |u_{1,2}\rangle\}$, where
\begin{eqnarray*}
|u_{1,1}\rangle \equiv\begin{bmatrix}|2\rangle \\ 0 \end{bmatrix},\quad |u_{1,2}\rangle \equiv\begin{bmatrix} |1\rangle  \\ (2+\mu/t)|2\rangle\end{bmatrix}, \quad |1\rangle,|2\rangle \in \mathds{C}^2,
\end{eqnarray*}
which correspond to the exponential solution $|\psi_{1,1}\rangle= |1,1\rangle|2\rangle$ and the power-law solution 
$|\psi_{1,2}\rangle = |1, 1\rangle |1\rangle +
(2+\mu/t) |1, 2\rangle |2\rangle $
of the bulk equation, in the notation of Eq. (\ref{ansatz}).
The boundary matrix computed in the first step of Algorithm \ref{algo:bound} now reads 
\begin{eqnarray*}
B = 2\begin{bmatrix}
-t-\mu & 0	& t & t \\
t & 0 & -t & -t \\
0 & -t-\mu & -t & -(2N+1)t-(N+1)\mu  \\
0 & -t & -t & -(2N+1)t-(N-1)\mu 
\end{bmatrix}, 
\end{eqnarray*}
which has a non-trivial kernel only if the parameter $\mu,t$ satisfy 
\begin{eqnarray*}
2N t + (N+1) \mu=0.
\end{eqnarray*}
Then, the corresponding kernel of $H_N-\epsilon$, or eigenspace of $H_N$ corresponding to 
eigenvalue $\epsilon=\mu+2t$, is spanned by 
\begin{eqnarray*}
|\epsilon\rangle = |\psi_{1,1}\rangle - |\psi_{1,2}\rangle = 
|1,1 \rangle (|2\rangle - |1\rangle) - (2+\mu/t) |1,2\rangle |2\rangle .
\end{eqnarray*} 
Note that, while as in the case of generic $\epsilon$, the eigenvector has contributions {\em only} from extended 
solutions, a power-law solution $|\psi_{1,2}\rangle$ now enters explicitly. 
Similar conclusions hold for other values of $\epsilon \in {\cal S}$.

\subsubsection*{$\bullet$ Majorana modes at the sweet spot, $t=\Delta$, $\mu=0$.}
For $\mu=0$, $H(z,z^{-1})-\epsilon$ can be verified to be {\em singular} for $\epsilon=\pm2t$, so that the 
kernel algorithm is inapplicable for those values. However, it is regular for all other values of $\epsilon$. 
We now diagonalize $H_N$ for $\epsilon \ne \pm 2t$ using the kernel algorithm. Since the characteristic equation in this case has no non-zero roots, there are {\em no} solutions with extended support. 
Therefore, we only need to find the kernels of the matrices $K_\epsilon^-$ and $K_\epsilon^+$, given by
\begin{eqnarray*}
\label{kitaevK2}
\hspace*{-10mm}
K_\epsilon^{-}=
\begin{bmatrix} 
h_1^\dagger & h_0-\epsilon & h_1 & 0\\ 
0 & h_1^\dagger & h_0-\epsilon & h_1\\
0 & 0 & h_1^\dagger & h_0-\epsilon\\
0 & 0 & 0 & h_1^\dagger
\end{bmatrix}, \quad
K_\epsilon^{+}=
\begin{bmatrix} 
h_1 & 0 & 0 & 0\\ 
h_0-\epsilon & h_1 & 0 & 0\\
h_1^\dagger & h_0-\epsilon & h_1 & 0\\
0 & h_1^\dagger & h_0-\epsilon & h_1
\end{bmatrix}.
\end{eqnarray*}
They are found to be spanned by $\{|u_1^{-}\rangle,|u_2^{-}\rangle\}$ and $\{|u_1^{+}\rangle,|u_2^{+}\rangle\}$, respectively, where
\begin{eqnarray*}
\hspace*{-1cm}
|u_1^{-}\rangle \equiv\begin{bmatrix} |-\rangle \\ 0 \\ 0 \\ 0\end{bmatrix}, \quad |u_2^{-}\rangle \equiv\begin{bmatrix} -\epsilon|+\rangle \\ 2t|-\rangle \\ 0 \\ 0\end{bmatrix}, \quad |u_1^{+}\rangle \equiv \begin{bmatrix} 0 \\ 0 \\ 0 \\ |+\rangle\end{bmatrix}, \quad |u_2^{+}\rangle \equiv \begin{bmatrix} 0 \\ 0 \\ 2t|+\rangle\\ -\epsilon|-\rangle \end{bmatrix},
\end{eqnarray*}
corresponding to the emergent solutions 
\begin{eqnarray*}
|\psi_1^-\rangle = |1\rangle|-\rangle, \quad|\psi_2^-\rangle = -\epsilon|1\rangle|+\rangle+ 2t|2\rangle|-\rangle, \\
|\psi_1^+\rangle = |N\rangle|+\rangle, \quad |\psi_2^+\rangle = -\epsilon|N\rangle|-\rangle+ 2t|N-1\rangle|+\rangle.
\end{eqnarray*}
In the first step of Algorithm \ref{algo:bound}, $B(\epsilon)$ is constructed using this as the basis, 
yielding
\begin{eqnarray*}
B (\epsilon) = \begin{bmatrix}
-\epsilon & \epsilon^2 -4 t^2 & 0 & 0 \\
 \epsilon & \epsilon^2 -4 t^2 & 0 & 0 \\
0 & 0 & -\epsilon & \epsilon^2 - 4 t^2 \\
0 & 0 & -\epsilon & - \epsilon^2 - 4 t^2
\end{bmatrix}.
\end{eqnarray*}
Since we assumed $\epsilon\ne\pm2t$, $B(\epsilon)$ has 
a non-trivial kernel only if 
$\epsilon=0$, in which case it is spanned by 
$\{|\alpha_1\rangle, |\alpha_2\rangle\}$, where
$|\alpha_1\rangle \equiv [1 \; 0 \; 0\; 0]^{\rm T},$ $|\alpha_2\rangle \equiv [0 \; 0\; 1\; 0]^{\rm T}. $
These represent the basis vectors of the zero-energy eigenspace 
of $H_N$,
$|\epsilon_1\rangle = |1\rangle|-\rangle,$ and  $|\epsilon_2\rangle 
= |N\rangle|+\rangle.$
They correspond to Majorana excitations, that have zero energy 
eigenvalue, and are perfectly localized only on the first or the last 
fermionic degrees of freedom in the chain. 

While, as noted, the kernel algorithm is inapplicable, one may 
verify that the remaining $2N-2$ eigenvectors of $H_N$ at the 
sweet spot belong to the eigenspace corresponding to $\epsilon=\pm 2t$. 
These eigenvectors, which arise from the countable kernel of $\H$ 
according to Theorem \ref{thyglorywithf}, are also perfectly 
localized, but in the bulk.

\section{Summary and Outlook}

Corner-modified BBT matrices describe a very large class
of tractable yet realistic boundary value problems of physics 
and engineering, from tight-binding models of fermions and
bosons to linear discrete-time dynamical systems. In this paper, 
we have characterized the spectral properties of corner-modified
matrices by purely {\em algebraic methods}, and have provided algorithms 
for exactly solving the eigensystem problem of {\em large and regular} 
matrices in this class. The regularity condition need not be 
a serious restriction in practice. In the language  of condensed-matter physics, 
it means that the single-particle Hamiltonian (the 
corner-modified BBT in question) displays no dispersionless energy 
bands. Nonetheless, with minor modifications, our algorithms do apply to 
singular BBT matrices as input, except, the output will be correct 
but not necessarily complete: some eigenvalues and/or generalized 
eigenvectors may be missed. 
Remarkably, our approach allows for all banded Toeplitz matrices, 
block or non-block, to be treated on an equal footing, 
without requiring the underlying matrix to be Hermitian nor 
excluding the strictly upper or lower triangular ones.
 
Our analysis of the eigensystem problem for corner-modified BBT 
matrices is unusual for its reliance on {\em symmetry}. 
Our starting point is to rewrite the eigenvalue equation as a system 
of two equations, the bulk and boundary equations, that we solve in 
succession. There are two types of solutions of the bulk equation.
One type of solution can be computed by reference to an
auxiliary BBL matrix. The resulting doubly-infinite matrices are 
{\em translation-invariant}, and so the associated eigensystem 
problem can be solved by a symmetry analysis. The latter, 
however, is highly unconventional from a quantum-mechanical
perspective, because the representation of the translation
symmetry is {\it not} unitary. As a consequence, the {\em extended 
solutions} of the bulk equation obtained via a BBL matrix can 
display three (and only three) possible behaviors: oscillatory, exponential 
decay, or exponential decay with a power-law correction. In addition, there
may exist {\em emergent solutions} of the bulk equation with {\em finite 
support}, localized near in the top and bottom entries of the 
eigenvector. Their relationship to translation symmetry is also 
striking, although less direct: they belong to the generalized 
kernel of a {\em truncated} translation symmetry. 
Finally, the boundary 
equation takes the solutions of the bulk equation as input, in order to  
select linear combinations that are the actual (possibly generalized) 
eigenvectors of the corner-modified BBT matrix of interest. While we have presented 
some exact results for the boundary equation as well, the latter need not be associated 
to a structured matrix and so a closed-form solution is 
not available in general.  
Notwithstanding, from a numerical efficiency standpoint, the key
observation is that the complexity of solving the bulk equation, and
also that of solving the boundary equation in the practically important case of 
symmetrical corner modifications, is {\em independent of the size} 
of the input matrix under consideration. 

In hindsight, one of our contributions can be interpreted, in physical parlance, as 
a generalization of the well-known Bloch's theorem for single-particle eigenfunctions 
to a class of boundary conditions 
not restricted to the standard periodic case,  with Eq. (\ref{ansatz}) providing a 
``generalized Bloch Ansatz'' \cite{prb1}. Interestingly, at least for the Hermitian case, 
the situation is reminiscent in many ways to the technique known as algebraic Bethe Ansatz, in the 
sense that one may solve for the eigensystem by finding the roots of associated 
polynomial equations, as opposed to the usual Bethe equations \cite{sutherland}. 

A number of promising directions for future research are prompted by 
our present investigation. From the point of view of applied mathematics,
it would be interesting to extend our approach to {\em multilevel} 
corner-modified BBT matrices: roughly speaking, these may be associated to 
sums of tensor products of our corner-modified BBT matrices; or, physically, 
to tight-binding models that cannot be reduced to one 
dimension. While much of our formalism goes through in higher dimensions, 
one conspicuous obstruction stems from the fact that there is no known 
equivalent of the Smith factorization for multivariate matrix 
Laurent polynomials \cite{krish85}, to the best of our knowledge. 
From the point of view of condensed-matter physics, 
{this work was prompted by the quest for exactly  
characterizing localized boundary modes and, more broadly, the   
role played by boundary conditions toward establishing a bulk/boundary
correspondence. On the one hand, this naturally prompts for the present 
mathematical tools to be applied to more general physical scenarios than  
equilibrium Hamiltonian systems as considered so far \cite{abc,prb1} -- 
including spectral properties of {\em non-equilibrium}, coherently or 
dissipatively driven fermionic matter, described 
by appropriate quadratic Floquet Hamiltonians 
\cite{EPL} or Lindblad master equations 
\cite{prosen,peter}.  
On the other hand, it is intriguing, as noted, that our
analysis brings to the fore translation symmetry, albeit 
in a non-unitary guise. Perhaps this is the starting point 
for formulating a symmetry principle behind the bulk/boundary
correspondence. But, what would be place of such a symmetry
principle in the light of the topological classification of
free fermions? }

\section*{Acknowledgements}

We acknowledge stimulating discussions with Cristiane Morais Smith, 
Bernard van Heck, and Shinsei Ryu.
Work at Dartmouth was supported in part by the National
Science Foundation through Grant No. PHY-1104403
and the Constance and Walter Burke Special Projects
Fund in Quantum Information Science.

\section{Appendix: Infinite banded block-Toeplitz transformations}
\label{ibbtwhynot}

In this appendix we will solve a physically motivated
problem associated to linear transformations of 
\(\V_{1,\infty}\equiv \P_{1,\infty}\V_d^S\)
of the form \(A=\P_{1,\infty}\A|_{\V_{1,\infty}}\).
The task is to compute the square-summable sequences
in \(\Ker A\), or some closely related corner-modified
version of \(A\). We will make this problem precise
after some preliminaries. 

Elements of \(\V_{1,\infty}\) can be seen as
half-infinite sequences. We will use the letter 
\(\Upsilon \in \V_{1,\infty}\) to denote one such sequence, 
and write 
\(\Upsilon \equiv \{|\upsilon_j\rangle\}_{j\in \N}\). 
If \(\A\) has bandwidth \((p,q)\), with \(p\leq q\), then \(A\) is 
induced by the ``infinite-downwards'' square array
\[  A=
\begin{bmatrix}
a_0 		& \dots  	& a_q   	&       		&					&0   				&\ddots		\\
\vdots  	&\ddots  	&\ddots 	&\ddots  	&					&            	&\ddots		\\ 
a_p   		&\ddots  	&\ \ a_0	&\ddots   &\ \ a_q	&		    			&						\\
       		&\ddots  	&\ddots	&\ddots	&	\ddots	&\ddots		&						\\
       		&          	&\ \ a_p 	&\ddots	&\ddots	&      			&						\\
0		  		&		  			&  				&\ddots	&					& 					&						\\
\ddots  	&\ddots	&					&					&					&						&
\end{bmatrix}.
\]
We will call \(A\) an {\it infinite} BBT matrix,
or IBBT for short. The transformation induced by $A$ is a IBBT transformation. 

\begin{defn}
\label{bbprojectors}
Let \(p' \equiv {\rm min}(p,0)\) and \(q' \equiv {\rm max}(0,q)\) for 
integers \(p\leq q\). The projector 
\[
P_B \Upsilon \equiv 
\{|\upsilon_j'\rangle\}_{j\in\N},\quad
|\upsilon_j'\rangle=\left \{
\begin{array}{lcl}
0& \mbox{if}& j=1,\dots,-p'\\
|\upsilon_j\rangle& \mbox{if} & -p'<j
\end{array}
\right., 
\]
is the {\it right bulk projector} for bandwidth \((p,q)\). 
The projector
\[  Q_B \Upsilon \equiv 
\{|\upsilon_j'\rangle\}_{j\in\N},\quad
|\upsilon_j'\rangle=\left \{
\begin{array}{lcl}
0& \mbox{if}& j=1,\dots,q'\\
|\upsilon_j\rangle& \mbox{if} & q'<j
\end{array}
\right., 
\]
is the {\it left bulk projector}. The corresponding left and 
right boundary projectors are \(P_\partial \equiv \mathds{1}-P_B\) and 
\(Q_\partial \equiv \mathds{1}-Q_\partial\), respectively.
\end{defn}
\noindent
With this definition, it follows that 
if \(p'=0\) (\(q'=0\)), then \(P_B=\mathds{1}\) (\(Q_B=\mathds{1}\)).

\begin{defn}
A linear transformation \(C\) of \(\V_{1,\infty}\) is a 
{\it corner-modified IBBT transformation} if there exists an 
IBBT transformation \(A=P_{1,\infty}\A|_{\V_{1,\infty}}\),
necessarily unique, such that \(P_BC=P_BA\). \(C\) is 
{\em symmetrical} if, in addition, \(CQ_B=AQ_B\). 
\end{defn}

\begin{lem}
\label{opkerinc}
If the principal coefficient \(a_{p'}\) of \(A\) is invertible, 
then \(\Ker P_BA=\P_{1,\infty}\Ker\A\). Otherwise, 
\(\P_{1,\infty}\Ker\A \subset \Ker P_BA\).  
\end{lem} 
\begin{proof}
See the proof of Thm.\,\ref{softy}.
In contrast to the situation for \(A_N\), the principal coefficient
\(a_{q'}\) plays no role here. As noted,  if \(p'=0\), then
\(P_B=\mathds{1}\).
\end{proof}

From here onwards, we denote the solution space of the bulk equation by
\(\mathcal{M} \equiv \Ker P_BA. \)
Since \(A\) is now a linear transformation of an infinite-dimensional 
vector space,  \(\mathcal{M}\) may also, in principle, be 
infinite-dimensional if  \(a_p'\) is not invertible. 
{We show next that  
it is finite-dimensional, although there is no guarantee that 
\(\dim \Range P_\partial\) matches \( \dim \mathcal{M}\). 
The proof of Theorem \ref{lemker} breaks down for IBBT transformations, 
and so one may expect the boundary matrix to be rectangular in general. }

Recall, from Thm.\,\ref{thm:semi},
that if \(\A\) is regular, then \(\mathcal{M}=\mathcal{F}_{1}^{-} \oplus \P_{1,\infty}\Ker\A \ ,
\) 
where $\mathcal{F}_1^{-} \in \V_{1,\sigma}$ 
for any $N>\tau+2\sigma$. If the principal coefficient \(a_{p'}\) of \(A\) 
is invertible, then \(\mathcal{F}_{1}^{-}=\{0\}\).  
The subspace 
\[  {\cal H} \equiv \Big\{\{|\upsilon_j\rangle\}_{j\in\N} \,|\,
\sum_{j=1}^\infty \langle\upsilon_j|\upsilon_j\rangle<\infty\Big\} \subset \V_{1,\infty}
\]
is the Hilbert space of square-summable 
sequences. We will denote square-summable sequences 
as \(|\Upsilon\rangle\in{\cal H}\), so that 
\(\langle \Upsilon_1|\Upsilon_2\rangle=\sum_{j\in \N}\langle \upsilon_{1,j}|
\upsilon_{2,j}\rangle <\infty\).
Our final task in this appendix is to compute a basis of  
\(\Ker C\cap {\cal H}\) for an arbitrary corner modified IBBT 
transformation. Physically, these states correspond to normalizable, 
bound states. 

\begin{lem}
\label{lemhilb}
Let $z \in \mathds{C}$, $z\ne 0$, and $s\in\mathds{N}$. 
If $|z|<1$, then $\P_{1,\infty}\mathcal{T}_{z,s}\subset\mathcal{H}$. 
\end{lem}
\begin{proof}
Let 
$\P_{1,\infty}\Phi_{z,v}|m\rangle \equiv \{j^{(v-1)}z^{j}|m\rangle\}_{j\in\N} 
= \{|\phi_j\rangle\}_{j\in\N}. $
The \(l^2\)-norm of this sequence would be 
given by \(
\sum_{j\in\N}\langle \phi_j|\phi_j\rangle = 
\sum_{j\in\N} (j^{(v-1)})^2|z|^{2j}  \),  
if convergent. The limit 
\[ \lim_{j\rightarrow \infty}
\left|\frac{((j+1)^{(v-1)})^2|z|^{2(j+1)}}{(j^{(v-1)})^2|z|^{2j}}\right|=|z|^2 \]
and so, by the ratio test, the series converges if \(|z|<1\) 
and diverges \(|z|>1\). It is immediate to check that it also diverges
if \(|z|=1\), in which case the series is attempting to sum a non-decreasing 
sequence of strictly positive numbers.
\end{proof}

\begin{thm}
If \(\A\) is regular, the space of the square-summable solutions of the bulk equation is 
given by
\[
\mathcal{M}\cap{\cal H}=
\mathcal{F}_{1}^{-}\oplus \P_{1,\infty}\bigoplus_{|z_\ell|<1}\, 
\Ker \A\cap \mathcal{T}_{z_\ell,s_\ell}.
\]
\end{thm} 
\begin{proof}
The sequences in \(\mathcal{F}_{1}^{-}\) have finite support, 
and so they are square-summable. Then, $\mathcal{F}_{1}^{-} \subset \mathcal{H}$ 
implies that \(
\mathcal{M}\cap{\cal H} = 
\mathcal{F}_{1}^{-} \oplus \big(\P_{1,\infty}\Ker\A\cap {\cal H}\big).
\)
For every $\Psi \in \Ker\A$, 
\begin{eqnarray*}
\P_{1,\infty}\Psi =
\sum_{\ell=1}^{n}\sum_{v=1}^{s_\ell}\sum_{m=1}^{d}\alpha_{\ell,v,m}
\P_{1,\infty}\Phi_{z_\ell,v}|m\rangle
\equiv
\{|\psi_j\rangle\}_{j\in\N}, 
\end{eqnarray*}
so that $|\psi_j\rangle = \sum_{m=1}^{d}|m\rangle\sum_{\ell=1}^{n}
y_{\ell,m}(j)z_{\ell}^{j},$ with 
$y_{\ell,m}(j)=\sum_{v=1}^{s_\ell}\alpha_{\ell,v,m}j^{(v-1)}$ 
polynomials in $j$ of degree at most $s_\ell$. The sequence 
$\P_{1,\infty}\Psi$ cannot be square-summable unless 
\(\lim_{j\rightarrow \infty}\langle \psi_j |\psi_j\rangle=0\), 
which in turn implies \(\lim_{j\rightarrow \infty}|\psi_j\rangle=0. 
\)
Hence, for any \(m=1,\dots,d\),   
\[
\lim_{j\rightarrow \infty}\langle m |\psi_j\rangle=0
\ \ \implies \ \ 
\lim_{j\rightarrow \infty}\sum_{\ell=1}^{n}y_{\ell,m}(j)z_\ell^j=0.
\]
The necessary condition for square-summability can be met if 
and only if $\alpha_{\ell,s,m}=0$ whenever \(|z_\ell|\geq 1\), for all $s,m$. 
\end{proof}

Based on the above characterization of the solution space $\mathcal{M}$, the 
Ansatz for the kernel vectors of $C$ in the space of square summable 
sequences may be written by suitably truncating the general 
Ansatz presented in Eq.\,\eqref{ansatz}.

\end{document}